\documentclass[twocolumn]{IEEEtran}

\usepackage{amssymb}
\usepackage{amsmath}
\usepackage{graphicx}
\usepackage{subfigure}
\usepackage{tabularx}
\usepackage{epsfig,epsf,color,balance,cite}
\newtheorem{theorem}{Theorem}
\newtheorem{remark}{Remark}

\begin{document}
\title{Hovering UAV-Based FSO Communications: Channel Modelling, Performance Analysis, and Parameter Optimization}

\author{Jin-Yuan Wang, \emph{Member, IEEE}, Yang Ma, Rong-Rong Lu, Jun-Bo Wang, \emph{Member, IEEE}, \\Min Lin, \emph{Member, IEEE}, and Julian Cheng, \emph{Senior Member, IEEE}
\thanks{Manuscript received XXXX; revised XXXX; accepted XXXX.
This work was supported in part by the Open Project of Shanghai Key Laboratory of Trustworthy Computing,
in part by the Opening Fundation of Key Laboratory of Opto-technology and Intelligent Control, Ministry of Education under Grant KFKT2020-06,
in part by the National Key Research and Development Program under Grant 2018YFB1801905,
in part by the National Natural Science Foundation of China under Grants 61960206005 \& 61960206006,
in part by the Jiangsu Province Basic Research Project under Grant BK20192002,
and in part by the Key International Cooperation Research Project under Grant 61720106003. (\emph{Corresponding authors: Jin-Yuan Wang, Jun-Bo Wang})}
\thanks{Part of the paper has been presented at \emph{IEEE International Conference on Communications}, Dublin, Ireland, 2020. \cite{ConfVer}.}
\thanks{Jin-Yuan Wang is with College of Telecommunications and Information Engineering, Nanjing University of Posts and Telecommunications, Nanjing 210003, China,
with Shanghai Key Laboratory of Trustworthy Computing, East China Normal University, Shanghai 200062, China,
and also with Key Laboratory of Opto-technology and Intelligent Control, Lanzhou Jiaotong University, Lanzhou 730070, China. (e-mail: jywang@njupt.edu.cn)}
\thanks{Yang Ma and Jun-Bo Wang are with National Mobile Communications Research Laboratory, Southeast University, Nanjing 210096, China. (e-mail: yangma@seu.edu.cn, jbwang@seu.edu.cn)}
\thanks{Rong-Rong Lu and Min Lin are with College of Telecommunications and Information Engineering, Nanjing University of Posts and Telecommunications, Nanjing 210003, China. (e-mail: 1019010205@njupt.edu.cn, linmin@njupt.edu.cn)}
\thanks{Julian Cheng is with the School of Engineering, The University of British Columbia, Kelowna, BC V1V 1V7, Canada. (e-mail: julian.cheng@ubc.ca)}
\thanks{Color versions of one or more of the figures in this article are available online at XXXX.}
\thanks{Digital Object Identifier: XXXX}
}

\maketitle
\begin{abstract}
Relay-assisted free-space optical (FSO) communication systems are exploited as a means to mitigate the limiting effects of the turbulence induced atmospheric scintillation.
However, conventional ground relays are stationary,
and their optimal placement is not always feasible.
Due to their mobility and flexibility,
unmanned aerial vehicles (UAVs) provide new opportunities for FSO relaying systems.
In this paper, a hovering UAV-based serial FSO decode-and-forward relaying system is investigated.
In the channel modelling for such a system, four types of impairments (i.e., atmospheric loss, atmospheric turbulence, pointing error, and link interruption due to angle-of-arrival fluctuation) are considered.
Based on the proposed channel model,
a tractable expression for the probability density function of the total channel gain is obtained.
Closed-form expressions of the link outage probability and end-to-end outage probability are derived.
Asymptotic outage performance bounds for each link and the overall system are also presented to reveal insights into the impacts of different impairments.
To improve system performance, we optimize the beam width, field-of-view and UAVs' locations.
Numerical results show that the derived theoretical expressions are accurate to evaluate the outage performance of the system.
Moreover, the proposed optimization schemes are efficient and can improve performance significantly.
\end{abstract}


\IEEEpeerreviewmaketitle

\section{Introduction}
\label{section1}
Recently, free-space optical (FSO) communications has attracted considerable attention to overcome the spectrum congestion problem \cite{FSO1,a1}, because unlike the radio frequency (RF) wireless communications,
FSO communications are unlicensed, directional,
immune to electromagnetic interference, and not easily intercepted.
Nevertheless, it is widely acknowledged that the terrestrial FSO signal is impaired by three important factors including atmospheric loss, atmospheric turbulence, and pointing error, which are all distance-dependent \cite{FSO10}.
Moreover, the transceivers in FSO communications are limited by the strict requirement of line-of-sight (LoS) alignment \cite{new07}.
These disadvantages lead to the development of the relay-assisted FSO communications \cite{FSO11},
whereby relays are placed properly between the source and the destination to improve system performance and reliability.
However, affected by the obstacles (such as lakes, mountains and buildings), the optimal positions to deploy relays may not be always feasible, and thus novel relaying schemes should be exploited.

To revolutionize the commonly-employed relaying network architectures, the unmanned aerial vehicle (UAV) based relaying scheme was proposed for FSO communications \cite{FSO2}.
Compared with the conventional terrestrial relays,
UAV-based relays have the intrinsic advantages of finding better communication environment and establishing the LoS links by adjusting positions dynamically.
With the perfect combination of FSO communications and UAVs,
the UAV-based FSO communication system is currently gaining significant attention,
which is regarded as a promising technology in many fields,
such as emergency response and military operation \cite{new10,new00,dd}.
It is expected that employing UAVs in FSO communication systems will inspire promising and innovative applications for future communication systems.

For UAV-based FSO relaying systems,
accurate channel modeling, tractable performance indicator expressions and efficient parameters optimization methods are of vital importance.
Several analytical models for hovering UAV based FSO communications were proposed \cite{performance1}.
In \cite{multielement}, the alignment and stability analyses for inter-UAV communications were presented.
In \cite{ESRPaper}, the ergodic sum rate for a UAV-based relay network with mixed RF/FSO channel was derived, and the joint effects of the atmospheric loss, the atmospheric turbulence, and the geometric and misalignment loss on the FSO channel were considered.
In \cite{performance2}, the throughput for a UAV-based mixed FSO/RF system with a buffer constraint was analyzed.
Note that the effects of the angle-of-arrival (AoA) fluctuations of due to orientation deviations of hovering UAVs were not considered in \cite{performance1}-\cite{performance2}.
By considering the non-orthogonality of the laser beam and the random fluctuations of the UAV's orientation and position, a novel channel model was proposed for UAV-based FSO communications \cite{statistical}.
In \cite{FSO3}, the FSO channel impaired by weak turbulence, pointing error and AoA fluctuation was established.
By considering the AoA fluctuation and pointing error,
the outage performance of an FSO communication system operating on high-altitude airborne platforms was investigated \cite{FSO4}.
Recently, by jointly considering the effects of atmospheric attenuation, atmospheric turbulence, pointing error and link interruption due to AoA fluctuation,
a novel statistical channel model for hovering UAV-based FSO communications was derived \cite{FSO5}.
In \cite{NonezeroBoresight}, by taking into account the effect of nonzero boresight pointing error, an extended model is established for UAV-based FSO communications.
However, due to the large number of complex functions and operators, the derived expressions \cite{FSO5,NonezeroBoresight} are cumbersome, and intuitive insights cannot easily be obtained for system design.
To facilitate the parameter optimization, we are motivated to develop tractable channel models and theoretical expressions.

For parameter optimizations, there exist two fundamental tradeoffs in UAV-based FSO communications \cite{FSO4,new11,FSO3,FSO8}.
First, the beam width balances the fading of pointing error and the strength of received signal.
That is, increasing beam width will mitigate the impact of pointing error,
but it also reduces the strength of the received signal.
Second, the selection of  field-of-view (FoV) involves a tradeoff between AoA fluctuation and received noise.
In other words, increasing the value of FoV can reduce the impact of AoA fluctuation but increase the received noise as well.
For network optimizations, the deployment or trajectory of UAVs can be designed to take full advantage of their ultra-flexibilities.
Such topics have been studied in open literature.
In \cite{new04,new12,new06}, the performance indicators like system throughput, network coverage, energy efficiency, and task completion time were optimized to find the optimal static locations or trajectory of UAVs.
When designing the UAVs' trajectory, most prior solutions rely on simplified channel models that are based on either the assumption that LoS link always exists or the LoS statistic model.
However, in reality, the LoS link may be obstructed by some obstacles (such as mountains and buildings). In other words, the LoS link may not always be available in practical systems. Therefore, the existence of the LoS link should be considered in the parameter optimization process \cite{FSO6}. Consequently, an important question arises:
How to efficiently find LoS link and obtain satisfied system performance for UAV-based FSO communication systems with obstacles?
Under the obstacle scenario, the UAV's location was optimized within a feasible region bounded by the minimum accepted elevation angles of source and destination \cite{addpaper1}. Note that the AoA fluctuation in \cite{addpaper1} is ignored, the considered lognormal distributed turbulence is only suitable for weak turbulence, and the feasible region limits the deployment of UAVs. Therefore, the parameter optimization problem in UAV-based FSO communications should be further investigated.

Motivated by the aforementioned discussions,
this paper focuses on a hovering UAV-based serial relaying FSO communication system,
and further investigates the channel modeling, performance analysis and parameter optimization problems.
The contributions of this paper are summarized as follows:
\begin{itemize}
  \item
  \emph{Channel modeling}: The considered FSO system involves three kinds of links, i.e., the ground-to-UAV (GU) link, the UAV-to-UAV (UU) link, and the UAV-to-ground (UG) link. For each link, the effects of atmospheric loss, atmospheric turbulence, pointing error and link interruption due to AoA fluctuation are jointly considered. Different from \cite{FSO5}, we consider the independence of four impairments, and obtain a simple and tractable expression for the probability density function (PDF) of the total channel gain, which can be effectively used for performance analysis. The accuracy of the derived PDF is verified in Section \ref{section5}.
  \item
  \emph{Performance analysis}: Based on the derived PDF of the total channel gain, the closed-form expressions of the link outage probability and the end-to-end outage probability are derived. Through asymptotic analysis, the asymptotic bounds of the link outage probability and the end-to-end outage probability are then obtained. Furthermore, insights about the derived expressions are also provided. Numerical results verify the accuracy of the derived theoretical expressions,  and these expressions can be directly utilized to evaluate system performance rapidly without time-intensive simulations.
  \item
 \emph{Parameter optimization}: To further improve the system performance, the beam width, the FoV and the UAVs' locations are optimized. Specifically, the minimum beam width is derived first, and the actual beam width of each link should be adjusted to be larger than or equal to the minimum value.
 Then, the FoV is optimized by minimizing the end-to-end outage probability with the beam width constraint, and the asymptotically optimal FoV is derived by solving a nonlinear equation.
 Finally, under the obstacle scenario, the locations of UAVs are optimized by solving a minimization problem. Then, the problem is transformed to a min-max problem, which can be effectively solved by using the embedded function ``fmincon" in MATLAB \cite{FSO6}. Numerical results verify the effectiveness of the proposed parameter optimization schemes.
\end{itemize}

The rest of the paper is organized as follows.
Section \ref{section2} presents the system model.
In Section \ref{section3}, the exact outage probabilities and the corresponding asymptotic bounds for the UAV-based FSO relaying system are analyzed.
In Section \ref{section4}, the parameter optimization is further investigated.
Numerical results are presented in Section \ref{section5}.
Finally, Section \ref{section6} concludes the paper.

\emph{Notations:} Throughout this paper, regular font indicates a scalar.
$X\sim {\cal N} (\mu, \sigma^2)$ denotes that $X$ follows a Gaussian distribution having mean $\mu$ and variance $\sigma^2$;
$\Gamma(\cdot)$ denotes the Gamma function;
$K_n(\cdot)$ denotes the modified Bessel function of the second kind with order $n$ \cite{FSO9};
${\rm{erf}}(x) = \frac{2}{{\sqrt \pi }}\int_0^x {{e^{ - {t^2}}}\rm{d}}t $ denotes the error function;
$\delta ( \cdot )$ denotes the Dirac delta function;
$G_{p,q}^{m,n}[ \cdot ]$ is the Meijer's G-function \cite{FSO9};
and $\Pr(\cdot)$ denotes the probability of an event.

\section{System Model}
\label{section2}
As shown in Fig. \ref{fig1}, we consider a UAV-based multi-hop relaying FSO communication system,
which includes one source node, one destination node, and $N$ UAV-based serial relay nodes.
In the system, the transmit optical signal from the source node propagates through $N$ serial relay nodes before detection at the destination node.
Specifically, the source node is equipped with a laser diode (LD) and transmits the optical signal to the first relay node.
For each relay node, a photodiode (PD) is first employed to convert the received optical signal to electrical signal, the decode-and-forward (DF) relaying scheme is then implemented, and the recovered signal is finally converted to optical signal by using an LD and transmitted to the next node. At the detection node, a PD is utilized to perform the photoelectric conversion.
Different from the ground fixed relays, the relay nodes are assumed to be mounted on the UAVs.
Moreover, we assume that the UAVs hover at their own fixed positions,
and the orientations of the transceiver in each link are aligned with each other.
However, the instantaneous positions of the relay nodes and the orientations of the transceiver deviate from the mean values due to numerous random events related to UAV hovering or building swaying.

\begin{figure}
\centering
\includegraphics[width=8cm]{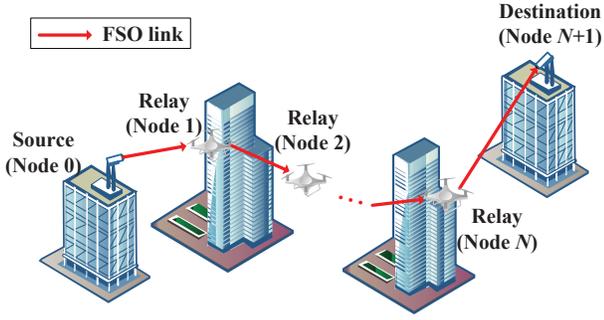}
\caption{A UAV-assisted serial relaying FSO communication system}
\label{fig1}
\end{figure}

\subsection{Received Signal Model}
\label{section2_1}
Here, the intensity modulation and direct detection is employed and the on-off keying (OOK) is used as the modulation scheme.
For simplification, the source node is denoted by node 0,
the relay nodes are denoted by nodes $1,2,\cdots,N$,
and the destination node is denoted by node $N+1$.
In the considered system, there are $N+1$ links, and the $i{\rm{th}}$ link is built by node $i-1$ and node $i$.
Mathematically, the received electrical signal at node $i$ in the $i{\rm{th}}$ link is given by
\begin{equation}
{y_i} = R{h_i}{x_i} + {n_i}, \; i= 1,\cdots, N+1,
\label{equ1}
\end{equation}
where $R$ is the optoelectronic conversion factor of the PD at node $i$;
$x_i$ is the transmitted optical signal of node $i-1$;
$h_i$ is the channel gain of the $i{\rm th}$ link,
and $n_i$ is the signal-independent additive white Gaussian noise having mean zero and variance $\sigma _{{\rm{n}},i}^2$, i.e., $n_i \sim {\cal N} (0, \sigma _{{\rm{n}},i}^2)$.
The transmitted optical signal is taken as symbols drawn equal-probably from the OOK constellation such that ${x_i} \in \{ 0,2{P_{\rm{t}}}\} $,
and $P_{\rm t}$ is the average transmitted optical power per link, which is related to the total transmit optical power $P$ by $P_{\rm{t}}=P/(N+1)$.
At each receiver, the background noise is assumed to be the dominant noise source, and the noise variance is a quadratic function with respect to the FoV of the receiver \cite{FSO5}, i.e.,
\begin{equation}
\sigma _{{\rm{n}}{\rm{,}}i}^{\rm{2}}= \Lambda \theta _{{\rm{FoV}}{\rm{,}}i}^{\rm{2}},
\label{plusFOV2}
\end{equation}
where $\Lambda$ is a coefficient related to the wavelength, bandwidth of the optical filter, spectral radiance and lens area.
$\theta _{{\rm{FoV}}{\rm{,}}i}$ is the receiver's FoV at node $i$.

\subsection{Channel Model}
\label{section2_2}
There are three types of links in the considered system, i.e.,
the GU link, the UU links, and the UG link.
Actually, the first link is the GU link,
the $(N+ 1){\rm th}$ link is the UG link,
and all other links (i.e., from the second link to the $N{\rm th}$ link) belong to the UU links.
For all the links, the total channel gain $h_i$  can be formulated as a product of four impairments, i.e.,
\begin{equation}
{h_i} = h_i^{{\rm{(l)}}}h_i^{({\rm{a}})}h_i^{({\rm{pe}})}h_i^{({\rm{aoa}})},\; i=1,\cdots, N+1,
\label{equ2}
\end{equation}
where $h_i^{{\rm{(l)}}}$ is the atmospheric loss, $h_i^{({\rm{a}})}$ is the atmospheric turbulence, $h_i^{({\rm{pe}})}$ is the pointing error, and $h_i^{({\rm{aoa}})}$ is the link interruption due to AoA fluctuation.
Inspired by the ideas in \cite{FSO4,FSO5} and \cite{FSO6}, three typical adjustable parameters (i.e., beam width, FoV, and link distance) are considered to be different in each link for the sake of the system parameter optimization,
and the other parameters are assumed to be the same in all $N+1$ links unless otherwise stated.
In the following, the statistical characteristic of the four impairments will be analyzed, respectively.

\subsubsection{Atmospheric loss}
Referring to \cite{plus1}, the atmospheric loss $h_i^{({\rm{l)}}} $  is determined by the exponential Beers-Lambert law as
\begin{equation}
h_i^{({\rm{l)}}} = \exp ( - {Z_i}\Phi ), \; i=1, \cdots, {N+1},
\label{equ3}
\end{equation}
where $Z_i$ is the distance of the $i{\rm{th}}$ link,
and $\Phi$ is the atmospheric attenuation coefficient related to visibility.

\subsubsection{Atmospheric turbulence}
For the atmospheric turbulence,
the Gamma-Gamma fading model is used here since its distribution is in close agreement with measurements under various turbulence conditions.
Therefore, the PDF of $h_i^{({\rm{a}})}$ is given by \cite{b1}
\begin{eqnarray}
{f_{h_i^{(\rm{a})}}}\!\!\left( {h_i^{({\rm{a}})}} \right) \!\!\!\! &=&\!\!\!\! \frac{{2{{(\alpha_i \beta_i )}^{\frac{\alpha_i  + \beta_i }{2}}}}}{{\Gamma (\alpha_i )\Gamma (\beta_i)}}{(h_i^{({\rm{a)}}})^{\frac{{\alpha_i  + \beta_i }}{2} - 1}}\nonumber\\
&\times&\!\!\!\! {K_{\alpha_i  - \beta_i }}\!\!\left( {2\sqrt {\alpha_i \beta_i h_i^{({\rm{a}})}} } \right),\; i\!=\!1,\cdots\!, N\!+\!1,
\label{equ4}
\end{eqnarray}
where the parameters ${\alpha_i}$ and ${\beta _i}$ related to large-scale and small-scale eddies are given by
\begin{equation}
\left\{
\begin{array}{l}
{\alpha _i} = {\left[ {\exp \left( \frac{0.49\sigma _{{\rm{R}},i}^2}{{{\left( {1 + 1.11\sigma _{{\rm{R}},i}^{12/5}} \right)}^{7/6}}} \right) - 1} \right]^{ - 1}}\\
{\beta _i} = {\left[ {\exp \left( \frac{0.51\sigma _{{\rm{R}},i}^2}{{{\left( {1 + 0.69\sigma _{{\rm{R}},i}^{12/5}} \right)}^{5/6}}} \right) - 1} \right]^{ - 1}}
\end{array}\right.,
\label{equ5}
\end{equation}
where $\sigma _{{\rm{R}},i}^2{\rm{ = }}1.23C_n^2{k^{7/6}}Z_i^{11/6}$ is the Rytov variance related to the link distance $Z_i$, $C_n^2$ is the index of refraction structure parameter of atmosphere, and $k = 2\pi /\lambda$ is the optical wave number with $\lambda$ being the wavelength.

\subsubsection{Pointing error}
Fig. \ref{fig2} shows the Gaussian beam footprint at receiver aperture, where the beam is orthogonal to the receiver lens plane. We consider a circular detection aperture of radius $r_a$ at node $i$, $i=1,\cdots,N+1$.
The radial displacement vector from the receiving aperture center to the beam center is expressed as ${{\bf{r}}_{{\rm{d}},i}} = [{x_{{\rm{d}},i}},{y_{{\rm{d}},i}}]$, which results from three different vectors:
i) the displacement vector induced by the transmitter's orientation deviation ${{\bf{r}}_{{\theta _{\rm{t}}},i}} = [{x_{{\theta _{\rm{t}}},i}},{y_{{\theta _{\rm{t}}},i}}]$,
ii) the displacement vector induced by the transmitter's position deviation ${{\bf{r}}_{{\rm{t}},i}} = [{x_{{\rm{t}},i}},{y_{{\rm{t}},i}}]$,
and iii) the displacement vector induced by the receiver's position deviation ${{\bf{r}}_{{\rm{r}},i}} = [{x_{{\rm{r}},i}},{y_{{\rm{r,}}i}}]$.
Consequently, the displacements located along the horizontal axe and the vertical axe at the detector plane are, respectively, given by
\begin{equation}
\left\{
\begin{array}{l}
{x_{{\rm{d}},i}} = {x_{{\rm{t}},i}}{\rm{ + }}{x_{{\rm{r}},i}} + {x_{{\theta _{\rm{t}}},i}}\\
 {y_{{\rm{d}},i}} = {y_{{\rm{t}},i}} + {y_{{\rm{r}},i}} + {y_{{\theta _{\rm{t}}},i}}
\end{array}\right..
\label{equ6}
\end{equation}
According to the central limit theorem, the position and orientation deviations follow Gaussian distributions as they result from numerous random events \cite{FSO5}.

\begin{figure}
\centering
\includegraphics[width=8cm]{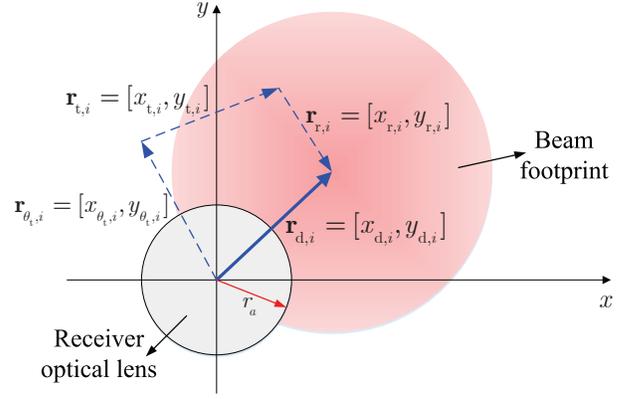}
\caption{The Gaussian beam footprint at receiver aperture for the $i$th link}
\label{fig2}
\end{figure}

\emph{a) UU links:} For these links, both ${x_{{\rm{t}},i}}$ and ${x_{{\rm{r}},i}}$ follow the same Gaussian distribution with mean zero and variance $\sigma _{{\rm{p,u}}}^{\rm{2}}$, i.e., ${x_{{\rm{t}},i}},{x_{{\rm{r}},i}} \sim {\cal N}(0, \sigma _{{\rm{p,u}}}^{\rm{2}})$.
Moreover, ${x_{{\theta _{\rm{t}}},i}} = {Z_i}\tan {\theta _{{\rm{tx}},i}} \simeq {Z_i}{\theta _{{\rm{tx}},i}}$,
where ${\theta _{{\rm{tx}},i}}\sim {\cal N}(0, \sigma _{{\rm{angle,u}}}^{\rm{2}})$. This indicates that ${x_{{\theta _{\rm{t}}},i}} \sim {\cal N} (0, Z_i^2 \sigma _{{\rm{angle,u}}}^{\rm{2}})$. Similarly, the distribution of ${y_{{\rm{d}},i}}$ can also be analyzed.
As a result, ${x_{{\rm{d}},i}}, {y_{{\rm{d}},i}} \sim {\cal N} (0, 2 \sigma_{\rm{p}, \rm{u}}^{2}+Z_{i}^{2} \sigma_{\rm{angle}, \rm{u}}^{2})$.

\emph{b) GU link: } For the first link, ${\theta _{{\rm{tx}},1}}$ is zero.
The variances of position deviation ${x_{\rm{t},1}}$ is assumed to be zero-mean Gaussian distribution with variance $\sigma _{{\rm{p,g}}}^{\rm{2}}$, i.e., ${x_{\rm{t},1}} \sim {\cal N} (0,\sigma _{{\rm{p,g}}}^{\rm{2}})$. Moreover, we can obtain ${x_{{\rm{r}},1}} \sim {\cal N} (0, \sigma _{{\rm{p,u}}}^{\rm{2}})$. Therefore, ${x_{{\rm{d}},1}}, {y_{{\rm{d}},1}} \sim {\cal N} (0,  \sigma_{\rm{p}, \rm{u}}^{2}+  \sigma_{\rm{p}, \rm{g}}^{2})$.

\emph{c) UG link: } Referring to the analysis for GU and UU links, we can obtain ${x_{{\rm t},N+1}} \sim {\cal N} (0,\sigma _{{\rm{p,g}}}^{\rm{2}})$, ${x_{{\rm{r}},N+1}} \sim {\cal N}(0, \sigma _{{\rm{p,u}}}^{\rm{2}})$, and ${x_{{\theta _{\rm{t}}},N+1}} \sim {\cal N} (0, Z_{N+1}^2 \sigma _{{\rm{angle,u}}}^{\rm{2}})$ for the $(N+1){\rm th}$ link. Consequently, ${x_{{\rm{d}},N+1}}, {y_{{\rm{d}},N+1}} \sim {\cal N} (0,  \sigma_{\rm{p}, \rm{u}}^{2}+  \sigma_{\rm{p}, \rm{g}}^{2} +Z_{N+1}^{2} \sigma_{\rm{angle}, \rm{u}}^{2} )$.

Assuming that the variables in (\ref{equ6}) are independent of each other \cite{FSO5},
the total radial displacement can be expressed as ${r_{{\rm{tr}},i}}{\rm{ = }}\sqrt {x_{{\rm{d}},i}^2 + y_{{\rm{d}},i}^2} $, which follows a Rayleigh distribution, i.e.,
\begin{equation}
{f_{{r_{{\rm{tr}},i}}}}({r_{{\rm{tr}},i}}) =\frac{r_{{\rm{tr}},i}}{\sigma _{{\rm{s}},i}^{\rm{2}}}\exp \left( { - \frac{r_{{\rm{tr}},i}^2}{2\sigma _{{\rm{s}},i}^{\rm{2}}} }\right),\; {r_{{\rm{tr}},i}} \geq0,
\label{pl1}
 \end{equation}
where the total displacement variances $\sigma_{\mathrm{s}, i}^{2}$ for different links are derived as
\begin{equation}
\label{plus2}
\sigma_{\mathrm{s}, i}^{2}\!=\!\left\{\!\!\!\begin{array}{ll}{\sigma_{\mathrm{p}, \mathrm{u}}^{2}+\sigma_{\mathrm{p}, \mathrm{g}}^{2},} & {i=1} \\
{2 \sigma_{\mathrm{p}, \mathrm{u}}^{2}+Z_{i}^{2} \sigma_{\mathrm{angle}, \mathrm{u}}^{2},} & {i=2, \ldots, N} \\ {\sigma_{\mathrm{p}, \mathrm{u}}^{2}+\sigma_{\mathrm{p}, \mathrm{g}}^{2}+Z_{N+1}^{2} \sigma_{\mathrm{angle}, \mathrm{u}}^{2},} & {i=N+1}\end{array}\right..
\end{equation}

Based on the radial displacement model in (\ref{pl1}), the PDF of $h_i^{({\rm{pe}})}$ is written as \cite{FSO8}
\begin{eqnarray}
{f_{h_i^{({\rm{pe}})}}}\left( {h_i^{({\rm{pe}})}} \right) \!\!\!\!&=&\!\!\!\! \frac{{\zeta _i^2}}{{A_i^{\zeta _i^2}}}{(h_i^{({\rm{pe}})})^{\zeta _i^2 - 1}},\nonumber\\
&&0 \le h_i^{({\rm{pe}})} \le {A_i},\; i\!=\!1,\cdots,N\!+\!1,
\label{equ8}
\end{eqnarray}
where ${A_i} = [{{\rm erf}} ({v_i})]^2$ is the fraction of the collected power when ${r_{{\rm{tr}},i}} = 0$,
${v_i} =  {\sqrt \pi  {r_{{a}}}} /\left( {\sqrt 2 {w_{z,i}}} \right)$ is the ratio between the aperture radius and the beam width. $\zeta _i^2 = {w_z}_{_{{\rm{eq}}},i}^2/{\rm{4}}\sigma _{{\rm{s}},i}^{\rm{2}}$ is the ratio between the squared equivalent beam width and displacement variance, where $\omega _{{z_{\rm eq }},i}^2 = \omega _{z,i}^2 \sqrt \pi  {\rm erf} (v_i)/(2 v_i e^{-v^2}) \approx \omega _{z,i}^2 + 3/(2\sqrt 2) $  is the equivalent beam width \cite{FSO7}.

\subsubsection{Link interruption due to AoA fluctuation}
As depicted in Fig. \ref{fig3}, due to the effect of AoA fluctuation, the beam is no longer orthogonal to the receiver plane. When an incident laser with $\theta_{{\rm a},i}$ arrives at the receiving plane,
the airy pattern may be out of detector range occasionally due to the relatively large orientation deviations of hovering UAV,
and the case that received AoA exceeds the range of FoV will result in an outage.
Here, the AoAs for different links are defined as \cite{FSO5}
\begin{equation}
\label{AOA}
\theta_{\mathrm{a}, i} \! \! \simeq \!\!\left\{\! \! \!\!\! \begin{array}{ll}
\sqrt{\theta_{\mathrm{rx}, i}^{2}+\theta_{\mathrm{ry}, i}^{2}}, &\!\! i=1 \\
\sqrt{\left(\theta_{\mathrm{tx}, i}\!+\!\theta_{\mathrm{rx}, i}\right)^{2}\!\!+\!\!\left(\theta_{\mathrm{ty}, i}\!+\!\theta_{\mathrm{ry}, i}\right)^{2}}, &\!\! i=2, \ldots, N \\
\sqrt{\theta_{\mathrm{tx}, i}^{2}+\theta_{\mathrm{ty}, i}^{2}}, &\!\! i=N+1
\end{array}\right.\!\!\!,
\end{equation}
where $\theta_{\mathrm{tx}, i}$ and $\theta_{\mathrm{ty}, i}$ are the horizontal and vertical misalignment orientations of the transmitter in the $i$th link.
Similarly, $\theta_{\mathrm{rx}, i}$ and $\theta_{\mathrm{ry}, i}$ represent the horizontal and vertical orientation deviations of the receiver.
Therefore, ${\theta _{{\rm{a}},i}},\; i=1,\cdots, N+1$ is Rayleigh-distributed, i.e.,
\begin{equation}
 {f_{{\theta _{{\rm{a}},i}}}}( {{\theta _{{\rm{a}},i}}} ) \!=\!\frac {\theta _{{\rm{a}},i}}{m_i\sigma _{{\rm{angle,u}}}^{\rm{2}} }\exp \!\!\left(\! { - \frac{\theta _{{\rm{a}},i}^2}{2m_i\sigma _{{\rm{angle,u}}}^{\rm{2}}}}\! \right), \theta _{{\rm{a}},i} \!\geq\! 0,
 \label{plus3}
\end{equation}
where $m_i$ varies with the type of links and is given by
\begin{equation}
m_i=\left\{\begin{array}{ll}{1,} & {i=\operatorname{l \ or} N+1} \\ {2,} & {i=2, \ldots, N }\end{array}\right. .
\label{pa1}
\end{equation}

\begin{figure}
\centering
\includegraphics[width=8cm]{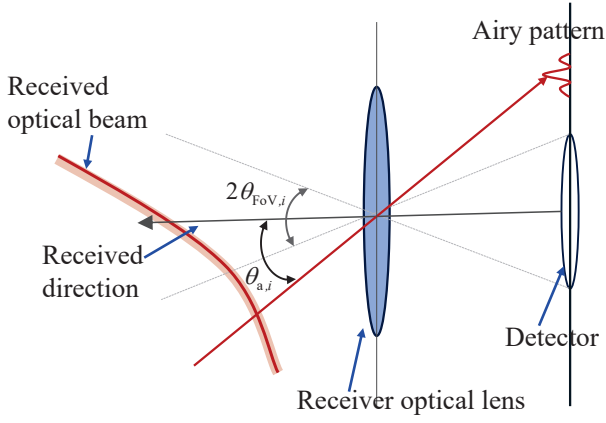}
\caption{A schematic diagram about the impact of AoA fluctuation for the $i$th link}
\label{fig3}
\end{figure}

We focus on the main lobe of the airy pattern, which has the most power.
The width of the main lobe is approximately equal to $2.4\lambda$, which is much smaller than the typical detector size \cite{d1}. Consequently, a reasonable consideration is the zero-one distribution that describes whether the incident laser is located on the receiving FoV or not \cite{add2}. Mathematically, given a fixed FoV, the link interruption (i.e., $h_i^{({\rm{aoa}})} = 0$) occurs if ${\theta _{{\rm a},i}} > {\theta _{{\rm{FoV}},i}}$, and the maximum signal power (i.e., $h_i^{({\rm{aoa}})} = 1$) is collected otherwise. Therefore, the corresponding PDF of $h_i^{({\rm{aoa}})}$ is given by
\begin{eqnarray}
{f_{h_i^{({\rm{aoa)}}}}}(h_i^{({\rm{aoa}})}) \!\!\!\!&=&\!\!\!\! \exp \left( { - \frac{{\theta _{{\rm{FoV}},i}^2}}{{2m_i\sigma _{{\rm{angle,u}}}^{\rm{2}}}}} \right)\delta (h_i^{({\rm{aoa}})}) \nonumber\\
 &+&\!\!\!\!\left[ {1 \!-\! \exp \left( { - \frac{{\theta _{{\rm{FoV}},i}^2}}{{2m_i\sigma _{{\rm{angle,u}}}^{\rm{2}}}}} \right)} \right]\delta (h_i^{{\rm{(aoa)}}} \!-\! 1),\nonumber\\
 &&\;\;\;\;\;\;\;\;\;\;\;\;\;\;\;\;\;\;\;\;\;\;\; i=1,\cdots, N+1.
\label{equ11}
\end{eqnarray}

\section{Outage Performance Analysis}
\label{section3}
In this section, the overall statistical characteristic of the FSO channel is analyzed.
Then, the link outage probability and the end-to-end outage probability are derived.
Finally, the asymptotic behaviors of the outage probabilities are studied.
Some insights are provided as well.

\subsection{Overall Channel Statistical Characteristic}
\label{section3_1}
For the four impairments in (\ref{equ2}), $h_i^{({\rm l})}$ and $h_i^{({\rm a})}$ are obviously independent and also independent of $h_i^{({\rm pe})}$ and $h_i^{({\rm aoa})}$. Although the orientation deviations of the transmitters (i.e., ${\theta _{{\rm{tx}},i}}$ and ${\theta _{{\rm{ty}},i}}$) contribute to the pointing error and the AoA fluctuation simultaneously, the correlation between $h_i^{({\rm pe})}$ and $h_i^{({\rm aoa})}$ is weak when ${\theta _{{\rm{FoV}},i}} \gg {\sigma _{{\rm{angle,u}}}}$.
The reason is provided as follows.
Note that ${\sigma _{{\rm{angle,u}}}}$ is the standard deviation of ${\theta _{{\rm{tx}},i}}$ and ${\theta _{{\rm{ty}},i}}$.
When ${\theta _{{\rm{FoV}},i}} \gg {\sigma _{{\rm{angle,u}}}}$,
compared with ${{\bf{r}}_{{\rm{t}},i}}$ and ${{\bf{r}}_{{\rm{r}},i}}$,
the impact of ${{\bf{r}}_{{\theta _{{\rm{t}},}}i}} \simeq [{Z_i}{\theta _{{\rm{tx}},i}},{Z_i}{\theta _{{\rm{ty}},i}}]$ in Fig. \ref{fig2} on pointing error is weak.
This indicates that the transmitter's orientation misalignment is weakly correlated with pointing error when ${\theta _{{\rm{FoV}},i}} \gg {\sigma _{{\rm{angle,u}}}}$.
Moreover, when ${\theta _{{\rm{FoV}},i}} \gg {\sigma _{{\rm{angle,u}}}}$, ${\theta _{{\rm{tx}},i}}$ and ${\theta _{{\rm{ty}},i}}$ have little effect on ${\theta _{{\rm{a}},i}}$, and thus the transmitter's orientation misalignment and the AoA fluctuation are also weakly correlated.
Consequently, the four impairments in (\ref{equ2}) can be approximated as independent variables.
Specially, for the GU link, the pointing error and the AoA fluctuation are no longer relevant since orientation deviations of the source node are approximately zero, and these four impairments are practically independent.

According to the above analysis, the PDF of ${h_i}$ is approximated as
\begin{equation}
f_{h_{i}}\left(h_{i}\right) \simeq \int_{0}^{+\infty} \frac{1}{h_{i}^{\prime}} f_{h_{i}^{(\rm{aoa})}}\!\!\!\left( \frac{h_{i}}{h_{i}^{\prime}}\right) f_{h_{i}^{\prime}}\left(h_{i}^{\prime}\right) \mathrm{d} h_{i}^{\prime},
\label{equ12}
\end{equation}
where $h_i^{'} \triangleq {h_i^{({\rm{l}})}h_i^{({\rm{a}})}h_i^{({\rm{pe}})}}$.
By solving (\ref{equ12}), we derive the PDF of $h_i$ and state the result in the following theorem.
\begin{theorem}
For the UAV-based FSO communication system, the PDF of the overall channel gain $h_i$ is approximated as
\begin{eqnarray}
{f_{{h_i}}}({h_i}) &\simeq&  \exp\left( { - \frac{{\theta _{{\rm{FoV}},i}^2}}{{2m_i\sigma _{{\rm{angle,u}}}^{\rm{2}}}}} \right)\delta ({h_i}) \nonumber\\
&+& \left[\! {1 \!- \!\exp\! \left(\! \!{ - \frac{{\theta _{{\rm{FoV}},i}^2}}{{2m_i\sigma _{{\rm{angle,u}}}^{\rm{2}}}}}\! \right)}\!\! \right]\frac{{{\alpha _i}{\beta _i}\zeta _i^2}}{{{A_i}h_i^{({\rm{l}})}\Gamma ({\alpha _i})\Gamma ({\beta _i})}} \nonumber \\
&\times& G_{1,3}^{3,0}\!\!\left[\! {\frac{{{\alpha _i}{\beta _i}}}{{{A_i}h_i^{{\rm{(l}})}}}{h_i} \left|\!\! {\begin{array}{*{20}{c}}
{\zeta _i^2}\\
{\zeta _i^2 - 1,{\alpha _i} - 1,{\beta _i} - 1}
\end{array}} \right.}\!\!\!\! \right]\!\!.
\label{equ13}
\end{eqnarray}
\label{the0}
\end{theorem}

\begin{proof}
See Appendix \ref{appa}.
\end{proof}

To verify the accuracy of the approximated PDF expression in \emph{Theorem \ref{the0}}, we provide Fig. \ref{fig4} in Section \ref{section5}. As can be seen, the approximate PDF in (\ref{equ13}) is tight as long as ${\theta _{{\rm{FoV}},i}} \ge 5{\sigma _{{\rm{angle,u}}}}$, which justifies the proposed approximation.

\subsection{Link Outage Probability}
\subsubsection{Exact expression}
The outage probability is an important metric to evaluate the performance of the FSO communication system.
The outage probability of the ${i{{\rm{th}}}}$ link ${p_i}$ is expressed as the probability that the instantaneous SNR ${\Upsilon _i}$ falls below the specified threshold ${\Upsilon _{{\rm{th}}}}$, i.e.,
\begin{equation}
{p_i} = \Pr ({\Upsilon _i} < {\Upsilon _{{\rm{th}}}})=\Pr({h_i} < {h_{{\rm{th}},i}}),
\label{equ16}
\end{equation}
where the corresponding  threshold of the channel gain is
\begin{equation}
{h_{{\rm{th,}}i}} = \sqrt {\frac{{{\Upsilon _{{\rm{th}}}}\sigma _{{\rm{n,}}i}^{\rm{2}}}}{2R^2{P_{\rm{t}}^2}}}=\frac{\theta _{{\rm{FoV}}{\rm{,}}i}}{RP_{\rm t}}\sqrt {\frac{{{\Upsilon _{{\rm{th}}}}\Lambda }}{2}}.
\label{equ17}
\end{equation}
The exact expression of (\ref{equ16}) is provided in the following theorem.

\begin{theorem}
For the UAV-based FSO communication system, the exact expression of the link outage probability  $p_i$, $i = 1,...,{N}{\rm{ + }}1$ is given by
\begin{eqnarray}
{p_i} \!\!\!\!\!&=& \!\!\!\!\! \exp \left(\!\! { - \frac{{\theta _{{\rm{FoV}},i}^2}}{{2m_i\sigma _{{\rm{angle,u}}}^{\rm{2}}}}}\!\! \right)\!+\! \left[ {1\! -\! \exp \left(\!\! { - \frac{{\theta _{{\rm{FoV}},i}^2}}{{2m_i\sigma _{{\rm{angle,u}}}^{\rm{2}}}}} \!\!\right)} \right]\nonumber\\
&\times& \!\!\!\!\!\! \frac{{\zeta _i^2}}{{\Gamma ({\alpha _i})\Gamma ({\beta _i})}} G_{2,4}^{3,1}\!\!\left[ {\frac{{{\alpha _i}{\beta _i}}\theta _{{\rm{FoV}},i}}{{{A_i}h_i^{({\rm{l)}}}}RP_{\rm t}}\!\sqrt {\!\frac{{{\Upsilon _{{\rm{th}}}}\Lambda }}{2}}\!\!\left|\!\! {\begin{array}{*{20}{c}}
{1,\zeta _i^2 \!+\! 1}\\
{\zeta _i^2,{\alpha _i},{\beta _i},0}
\end{array}} \right.}\!\!\!\!\! \right]\!\!,
\label{equ18}
\end{eqnarray}
where $m_i$ is given by (\ref{pa1}).
\label{the1}
\end{theorem}

\begin{proof}
See Appendix \ref{appb}.
\end{proof}

\subsubsection{Asymptotic bound}
For large transmit power ${P_{\rm{t}}}$, the asymptotic bound of the link outage probability is given in the following theorem.

\begin{theorem}
For the UAV-based FSO communication system, the asymptotic bound of the link outage probability ${p_{{\rm{bound}},i}}$, $i = 1,...,{N}{\rm{ + }}1$ for large transmit power ${P_{\rm{t}}}$ is given by
\begin{equation}
{p_{{\rm{bound}},i}} = \exp \left( { - \frac{{\theta _{{\rm{FoV}},i}^2}}{{2m_i\sigma _{{\rm{angle,u}}}^{\rm{2}}}}} \right).
\label{equ19}
\end{equation}
\label{the2}
\end{theorem}

\begin{proof}
See Appendix \ref{appc}.
\end{proof}

\begin{remark}
In \emph{Theorem \ref{the2}}, the asymptotic bound  (\ref{equ19}) varies with $\theta_{{\rm FoV},i}^2$ and $m_i$.
From (\ref{pa1}), it is known that $m_i=1$ for the UG/GU link, while $m_i=2$ for the UU links.
Given a fixed $\theta_{{\rm FoV},i}^2$ for all the links,  the asymptotic bound  (\ref{equ19}) for the UU links is larger than that for the GU/UG link, the asymptotic bound  (\ref{equ19}) for the GU link is the same as that for the UG link, which indicates that the UU links achieve the worst asymptotic outage performance, and the GU link and the UG link achieve the same asymptotic outage performance.
\label{rem1}
\end{remark}

\subsection{End-to-end Outage Probability}

\subsubsection{Exact expression}
For DF relaying, outage of each intermediate link may lead to the outage of the relaying system. Therefore, the end-to-end outage probability is given by
\begin{equation}
{P^{{\rm{out}}}} = 1 - \prod\limits_{i = 1}^{{N}{\rm{ + }}1} {\left( {1 - {p_i}} \right)}  .
\label{equ21}
\end{equation}
Then, the end-to-end outage probability is obtained in the following theorem.

\begin{theorem}
 For the UAV-based FSO communication system, the exact expression of the end-to-end outage probability ${P^{{\rm{out}}}}$ is expressed as
\begin{eqnarray}
\!\!\!\!\!\!\!{P^{{\rm{out}}}} \!\!\!\!\!\! &=& \!\!\!\!\! 1 \!\!-\!\! \left[ {1 \!-\! \exp \left( { - \frac{{\theta _{{\rm{FoV}},{\rm{1}}}^2}} {{{\rm{2}}\sigma _{{\rm{angle,u}}}^{\rm{2}}}}} \right)} \right]\nonumber \\
&  \times&\!\!\!\!\!\!\! \left[\! {1\! - \!\exp \!\left( \!\!{ - \frac{{\theta _{{\rm{FoV}},{N} + {\rm{1}}}^2}}{{{\rm{2}}\sigma _{{\rm{angle,u}}}^{\rm{2}}}}}\!\! \right)} \!\right]
 \prod\limits_{i = {\rm{2}}}^{{N}} {\left[ \!{{\rm{1}}\! -\! \exp\! \left(\! { - \frac{{\theta _{{\rm{FoV}},i}^2}}{{{\rm{4}}\sigma _{{\rm{angle,u}}}^{\rm{2}}}}}\! \right)}\! \right]} \nonumber \\
 &  \times&\!\!\!\!\!\!\!  \prod\limits_{i = 1}^{{N} \!+\! 1}\!\!\! {\left(\!\! {1\!\! -\! \frac{{\zeta _i^2}}{{\Gamma ({\alpha _i})\Gamma ({\beta _i})}}G_{2,4}^{3,1}\!\!\left[\! {\frac{{{\alpha _i}{\beta _i}}}{{{A_i}h_i^{({\rm{l)}}}}}{h_{{\rm{th,}}i}}\!\!\left|\!\! {\begin{array}{*{20}{c}}
{1,\zeta _i^2 \!+\! 1}\\
{\zeta _i^2,{\alpha _i},{\beta _i},0}
\end{array}} \right.} \!\!\!\!\!\right]} \!\right)}\! .
\label{equ22}
\end{eqnarray}
\label{the3}
\end{theorem}

\begin{proof}
The proof is straightforward by submitting \eqref{equ18} into \eqref{equ21}.
\end{proof}

\subsubsection{Asymptotic bound}
Due to the impact of the intermediate links performance bound in \eqref{equ19}, the corresponding system performance is also limited by AoA fluctuation when the transmitted power is large. The asymptotic bound for the end-to-end outage probability is given in the following theorem.

\begin{theorem}
For the UAV-based FSO communication system, the asymptotic bound of the end-to-end outage probability $P_{\rm{bound}}^{{\rm{out}}}$ for large transmit power ${P_{\rm{t}}}$ is given by
\begin{eqnarray}
\!\!\!\!\!\!\!P_{\rm{bound}}^{{\rm{out}}} \!\!\!\!\! &=& \!\!\!\!\! 1 - \left[ {1 - \exp \left( { - \frac{{\theta _{{\rm{FoV}},{\rm{1}}}^2}}{{{\rm{2}}\sigma _{{\rm{angle,u}}}^{\rm{2}}}}} \right)} \right]\nonumber \\
& \times& \!\!\!\!\!\! \left[ \!{1\! -\! \exp\! \left( \!\!{ - \frac{{\theta _{{\rm{FoV}},{N} \!+\! {\rm{1}}}^2}}{{{\rm{2}}\sigma _{{\rm{angle,u}}}^{\rm{2}}}}}\!\! \right)}\! \right]\!\! \prod\limits_{i = {\rm{2}}}^{{N}} \!{\!\left[\! {{\rm{1}}\! -\! \exp\! \left(\!\! { - \frac{{\theta _{{\rm{FoV}},i}^2}}{{{\rm{4}}\sigma _{{\rm{angle,u}}}^{\rm{2}}}}} \!\!\right)}\!\! \right]}\!\!.
\label{equ23}
\end{eqnarray}
\label{the4}
\end{theorem}

\begin{proof}
The proof is straightforward by substituting \eqref{equ19} into \eqref{equ21}.
\end{proof}

\begin{remark}
It is worth mentioning that the asymptotic bound $P_{\rm{bound}}^{{\rm{out}}}$ in \eqref{equ23} deteriorates with the increase of the number of relays ${N}$.
Interestingly, increasing relay number worsens the outage performance when the system outage performance has attained the bound in \eqref{equ23}.
\label{rem2}
\end{remark}

\section{Parameter Optimizations}
\label{section4}
In this section, the key parameter optimizations for the UAV-based FSO communications are investigated.
Specifically, the beam width, the FoV, and  the UAV locations are optimized.

\subsection{Beam Width Adjustment}
Since hovering UAVs have larger deviations of the position and orientation than fixed ground platforms, pointing error in UAV-based links is more severe, hence the adjustment for beam width is especially important to improve the diversity order gain and outage performance for different communication conditions.

In (\ref{equ8}), the ratio between the squared equivalent beam width and displacement variance is $\zeta _i^2 = {w_z}_{_{{\rm{eq}}},i}^2/(4\sigma _{{\rm{s}},i}^2) \approx [\omega _{z,i}^2 + 3/(2\sqrt 2)]/(4\sigma _{{\rm{s}},i}^2)$. In (\ref{equ4}), $\beta_i$ denotes the parameters related to small-scale eddy in atmospheric turbulence.
According to \cite{FSO7} and \cite{FSO13}, it can be known that the condition $\zeta _i^2 < {\beta _i}$ represents that pointing error becomes dominant in relation to atmospheric turbulence in the ${i{{\rm{th}}}}$ link.
Since the turbulence parameter ${\beta _i}$, which is related to the atmospheric condition, cannot be chosen arbitrarily, adjusting beam width ${w_{z,i}}$ to satisfy $\zeta _i^2  \ge  {\beta _i}$ is an appropriate measure.
By letting $\zeta _i^2{\rm{ = }}{\beta _i}$, the minimum value of received beam width in the ${i{{\rm{th}}}}$ link is given by
\begin{equation}
\omega _{z,i}^{\min }{\rm{ = }}\sqrt {4\beta_i \sigma _{{\rm{s}},i}^2 - \frac{3}{2\sqrt 2}},
\label{p141}
\end{equation}
Here, to avoid the severe impairment of pointing error, the beam width of each link in the FSO relaying system should satisfy
\begin{equation}
\omega _{z,i} \ge \omega _{z,i}^{\min },\; i = 1,\cdots, N + 1.
\label{beam}
\end{equation}
It is noteworthy to mention that the power received at fixed-size detector reduces with an increase of the beam width.
Consequently, excessive increase of the beam width is also inadvisable.

\subsection{FoV Optimization With the Beam Width Constraint}
 Due to the asymptotic bound \eqref{equ23}, the achievable end-to-end outage probability is limited.
 Therefore, after determining the minimum beam width, the FoV optimization problem should be considered to minimize the end-to-end outage probability, i.e.,
\begin{align}
 &\quad \mathop {\min }\limits_{\theta_{{\rm{FoV}},1},...,\theta_{{\rm{FoV}},N+1}}  P^{{\rm{out}}} \nonumber \\
\rm{s.t.}
 &\quad \omega _{z,i} \ge \omega _{z,i}^{\min }, \ i = 1,...,{N} + 1.
\label{equ25}
\end{align}

Since FoV adjustment of one link does not affect the other links,
problem (\ref{equ25}) is reduced to optimizing FoV of each link individually, i.e.,
\begin{align}
 &\quad \mathop {\min }\limits_{\theta_{{\rm{FoV}},i}}  p_i \nonumber \\
\rm{s.t.}
 &\quad \omega _{z,i} \ge \omega _{z,i}^{\min }.
\label{equ25_1}
\end{align}

Considering the complexity of the exact link outage probability and the demand of practical FSO communications, we optimize the system performance for large $P_{\rm{t}}$.
By substituting \eqref{plusFOV2} and \eqref{equ17} into \eqref{equ20}, we obtain the link outage probability for large $P_{\rm{t}}$ with respect to $\theta_{{\rm{FoV}},i}$ as
 \begin{equation}
\begin{aligned}
{p_i} \simeq L(\theta _{{\rm{FoV}},i}^2) + \Theta \left[ {1 - L(\theta _{{\rm{FoV}},i}^2)} \right]{\left( {{\theta_{{\rm{FoV}},i}}} \right)^{{\beta _i}}},
\end{aligned}
\label{plusFOV3}
\end{equation}
where $L(\theta _{{\rm{FoV}},i}^2)$ and $\Theta$ are defined as
\begin{eqnarray}
\left\{ \begin{array}{l}
L(\theta _{{\rm{FoV}},i}^2) = \exp \left( { - \frac{{\theta _{{\rm{FoV}},i}^2}}{{2m_i\sigma _{{\rm{angle}},{\rm{u}}}^{\rm{2}}}}} \right)  \\
\Theta  = \frac{{\zeta _i^2\Gamma ({\alpha _i} - {\beta _i})}}{{\Gamma ({\alpha _i})\Gamma ({\beta _i})(\zeta _i^2 - {\beta _i}){\beta _i}}}{\left( {\frac{{{\alpha _i}{\beta _i}}}{{{A_i}h_i^{({\rm{l}})}R{P_{\rm{t}}}}}\sqrt {\frac{{{\Upsilon _{{\rm{th}}}}\Lambda }}{2}} } \right)^{{\beta _i}}}
\end{array} \right..
\end{eqnarray}

By taking derivative of \eqref{plusFOV3} with respect to $\theta _{{\rm{FoV}},i}$ and setting it to be zero, we obtain the following nonlinear equation
\begin{eqnarray}
&&\!\!\!\!\!\!\!\!m_i \Theta \sigma_{\rm{angle,u}}^{2} \beta_{i}\left(\theta_{{\rm{FoV}},i}\right)^{\beta_{i}-2}\left[1-L\left(\theta_{\mathrm{FoV}, i}^{2}\right)\right]  \nonumber\\
&&\;\;\;+\Theta\left(\theta_{{\rm{FoV}},i}\right)^{\beta_{i}} L\left(\theta_{\mathrm{FoV}_{i}}^{2}\right)-L\left(\theta_{\mathrm{FoV}, i}^{2}\right)=0.
\label{plusFOV1}
\end{eqnarray}
In this paper, the asymptotically optimal $\theta _{{\rm{FoV}},i}$ is updated to match different communication conditions by numerically solving \eqref{plusFOV1}.

\subsection{UAV Location Optimization}
To improve the reliability of the end-to-end link, the UAVs' positions can be optimized.
To avoid the undesirable bound of AoA fluctuation (i.e. eq. (\ref{equ19})),
the FoV of each receiver is typically set to be a large fixed value.
In this subsection, we assume that the FoV is sufficiently large so that the AoA fluctuation is neglectable.

Here, we consider an FSO relaying system with $N$ UAVs and $N_0$ obstacles.
To facilitate the optimization, all UAVs are assumed to be located at the same height.
For all UAVs, the two-dimensional coordinates in the XY plane are optimized.
Without loss of generality, the positions of UAVs are set to be $(x_i, y_i), i=1,\cdots, N$. Moreover, $(x_0, y_0)$ and $(x_{N+1}, y_{N+1})$ are the coordinates of source node and destination node, respectively.
Due to the existence of obstacles, the UAVs cannot be deployed on the positions of obstacles.
An indicator function ${f_n}({x_i},{y_i},{x_{i + 1}},{y_{i + 1}})$  is used to indicate whether the $i{\rm th}$ link is block by the $n{\rm th}$ obstacle. Specifically, if the $n{\rm th}$ obstacle blocks the $i{\rm th}$ link, ${f_n}({x_i},{y_i},{x_{i + 1}},{y_{i + 1}}) \geq 0$, otherwise ${f_n}({x_i},{y_i},{x_{i + 1}},{y_{i + 1}}) < 0$.
Here, the UAV location optimization problem can be formulated as
\begin{eqnarray}
&&\mathop {\min }\limits_{{x_{\rm{1}}},{y_{\rm{1}}},...,{x_N},{y_N}} {P^{{\rm{out}}}} \nonumber \\
{\rm{s}}{\rm{.t}}{\rm{.}}&&{f_n}({x_i},{y_i},{x_{i + 1}},{y_{i + 1}}) < 0, \nonumber \\
&&  i \in 0,1,\cdots,N;  n \in 1,2,\cdots,{N_0}.
\label{add1}
\end{eqnarray}

To solve the optimization problem,
we assume that the atmospheric turbulence is dominant in relation to the pointing error (i.e., $\zeta _i^2 > {\beta _i},i = 1,\cdots,N + 1$), and the FoV ${\theta _{{\rm{FoV}},i}}$ is large.
With these assumptions, the expression (\ref{equ20}) can be further written as
\begin{eqnarray}
\!\!\!\!\!\!{p_i}\!\!\!\!\!\! &\approx&\!\!\!\!\!\! \frac{{\zeta _i^2\Gamma ({\alpha _i} - {\beta _i})}}{{\Gamma ({\alpha _i})\Gamma ({\beta _i})(\zeta _i^2 - {\beta _i}){\beta _i}}}{\left( {\frac{{{\alpha _i}{\beta _i}}}{{{A_i}h_i^{({\rm{l}})}}}} \right)^{{\beta _i}}}{\left( {h_{{\rm{th,}}i}^{}} \right)^{{\beta _i}}}  \nonumber\\
\!\!\!&=&\!\!\!\!\!\!  \underbrace {\frac{{\zeta _i^2\Gamma ({\alpha _i} \!\!-\!\! {\beta _i})}}{{\Gamma ({\alpha _i})\Gamma ({\beta _i})(\zeta _i^2 \!\!-\!\! {\beta _i}){\beta _i}}}\!\!{\left(\!\!\! {\frac{{{\alpha _i}{\beta _i}}}{{{A_i}h_i^{({\rm{l}})}R}}\sqrt {\!\frac{{{\Upsilon _{{\rm{th}}}}\!\sigma _{{\rm{n,}}i}^{\rm{2}}}}{2}} } \right)^{\!\!\!{\beta _i}}}}_{\triangleq a_i} {\!\!\left( {{P_{\rm{t}}}} \right)^{\!\!{\rm{ - }}{\beta _i}}}\!\!.
\end{eqnarray}

\begin{figure*}
\centering
\includegraphics[width=18cm]{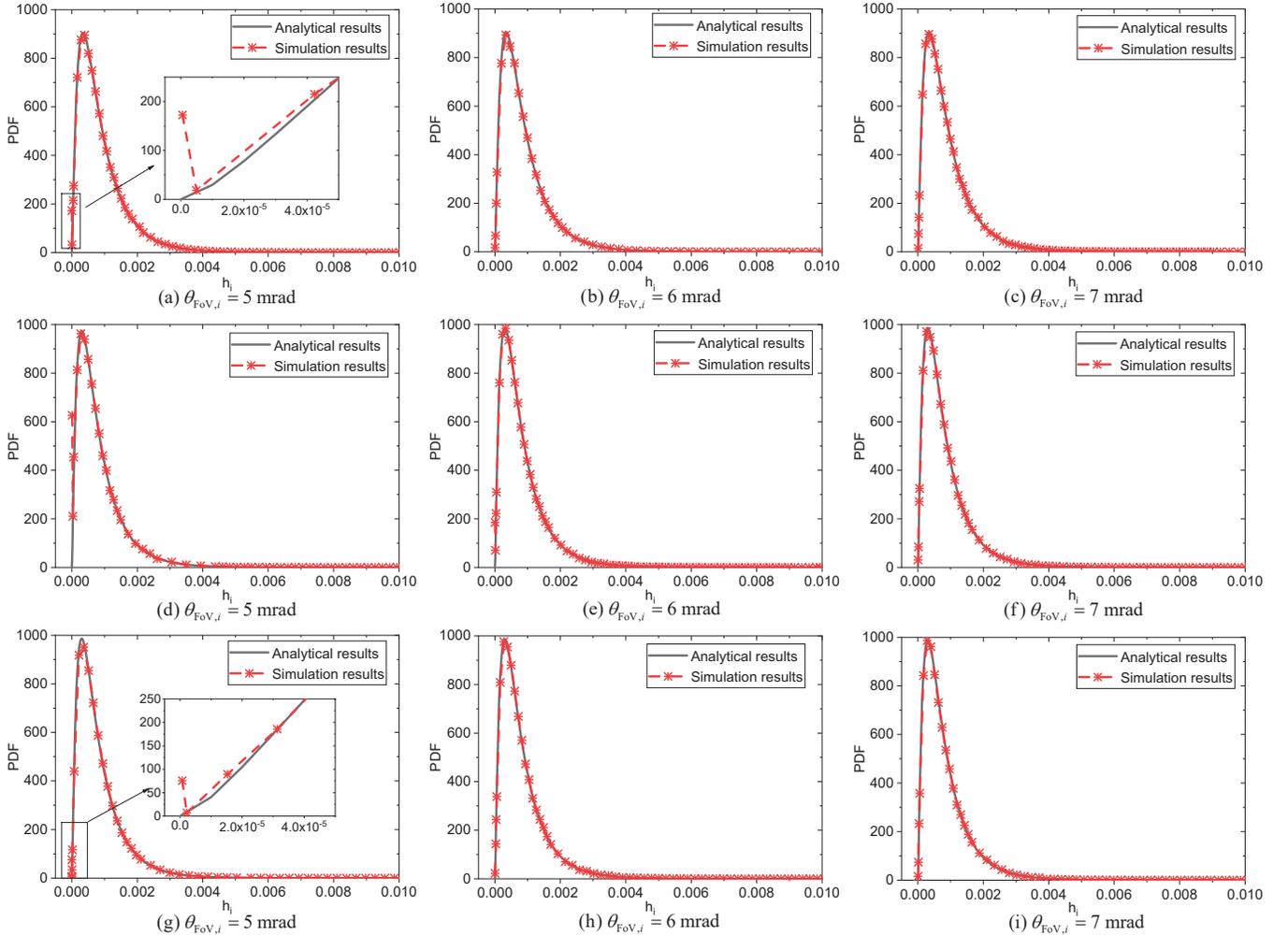}
\caption{Comparisons of the derived PDF (\ref{equ13}) and the simulated PDF with different $\theta_{{\rm{FoV}},i}$ for different links when $Z_i = 250$ m and $w_{z_i} = 2$ m, where (a)-(c) are for the GU link, (d)-(f) are for the UU links, (g)-(i) are for the UG link}
\label{fig4}
\end{figure*}

We assume that the $i_{{\rm{max}}}{{\rm{th}}}$ link is the longest with parameters $a_{i_{\max }}$ and $\beta_{i_{\max }}$.
Since the parameter ${\beta _i}$ decreases monotonically with an increase of the ${i{{\rm{th}}}}$ link distance, we have
\begin{eqnarray}
\mathop {\lim }\limits_{{P_{\rm{t}}} \to \infty } \frac{{1 - \prod\nolimits_{i = 1}^{N + 1} {\left( {1 - p_i} \right)} }}{{p_{{i_{{\rm{max}}}}}}} \!\!\!\!&=& \!\!\!\! \mathop {\lim }\limits_{{P_{\rm{t}}} \to \infty } \frac{{\sum\nolimits_{i = 1}^{N + 1} {p_i} }}{{p_{{i_{\max }}}}} \nonumber \\
 &=& \!\!\!\! \mathop {\lim }\limits_{{P_{\rm{t}}} \to \infty } \frac{{\sum\nolimits_{i = 1}^{N+ 1} {{a_i}{ P_{\rm{t}} ^{ - {\beta _i}}}} }}{{{a_{{i_{{\rm{max}}}}}}{P_{\rm{t}}^{ - {\beta _{{i_{\max }}}}}}}} \nonumber\\
 &=& \!\!\!\! 1,
  \label{MDL}
\end{eqnarray}
which indicates that the system outage probability is determined by the performance of the link having the maximum distance.
Therefore, for a large value of $P_{\rm t}$, problem (\ref{add1}) can be transformed into
\begin{eqnarray}
&&\mathop {\min }\limits_{{x_{\rm{1}}},{y_{\rm{1}}},\cdots,{x_N},{y_N}} \max \{ Z_1, Z_2,\cdots, Z_N\}\nonumber \\
{\rm{s}}{\rm{.t}}{\rm{.}}&& {f_n}({x_i},{y_i},{x_{i + 1}},{y_{i + 1}}) < 0\nonumber\\
 && i \in 0,1,\cdots,N; n \in 1,2,\cdots,{N_0},
 \label{add2}
\end{eqnarray}
where ${Z_i} = \sqrt {{{({x_i} - {x_{i + 1}})}^2} + {{({y_i} - {y_{i + 1}})}^2}} $ is the distance of the $i{\rm th}$ link.
To solve problem (\ref{add2}) effectively, the grouping optimization or multi-variate optimization can be used, which can be implemented by using the embedded function ``fmincon" in MATLAB \cite{FSO6}.

\section{Numerical Results}
\label{section5}
In this section, some numerical results will be presented to support the theoretical claims made in previous sections.
Unless stated otherwise, the default values of the parameters used for simulations are given by
$\lambda  = 1550$ nm,
$C_{n}^{2}=5\times10^{-14}\; {\rm m}^{-2 /3}$,
$R=0.9$,
$\Phi  = 1\; {\rm{km}}^{ - 1}$,
$\sigma _{\rm{n}}^2 = {6.4\times10^{ - 14}}$ corresponding to ${\theta _{{\rm{FoV}},i}} = 8$ mrad,
receiver lens radius ${r_a} = 5$ cm,
the standard deviation of UAV and ground displacement ${\sigma _{{\rm{p,u}}}} = {\sigma _{{\rm{p,g}}}} = 10$ cm,
the standard deviation of UAV orientation ${\sigma _{{\rm{angle,u}}}} = 1.2$ mrad,
and ${\gamma _{{\rm{th}}}} = 10$ dB \cite{FSO5}.

\subsection{Channel Modelling Results}
The derivation of the PDF in  (\ref{equ13}) is based on the assumption that the four types of impairments in (\ref{equ2}) are independent of each other. To verify the accuracy of (\ref{equ13}), we compare
the curves of (\ref{equ13}) with the simulated PDFs for different links when $Z_i = 250$ m and $w_{z_i} = 2$ m, as shown in Fig. \ref{fig4}. Note that Figs. \ref{fig4}(a)-(c) are for the GU link, Figs. \ref{fig4}(d)-(f) are for the UU links, and Figs. \ref{fig4}(g)-(i) are for the UG link.
For the simulated PDFs, the four impairments in (\ref{equ2}) are set to be correlated random variables.
As can be seen in Figs. \ref{fig4}(a), (d) and (g),
when $\theta_{{\rm{FoV}}, i}=5$ mrad, the differences between analytical results when $h_i = 0$ and simulation results are large. Specifically, for the GU, UU and UG links, $f_{h_i}(0)$ in \eqref{equ13} are 1.3, 1.7$\times10^{-4}$ and 1.7$\times 10^{-4}$, respectively.
But the simulated values at $h_i=0$ are 172.7, 626.6, and 75.6, respectively.
For other values of $h_i$ in Fig. \ref{fig4}(a), (d) and (g), the analytical results match the simulation results well.
This indicates that small FoV in each link will result in the mismatch between the analytical results and simulation results when $h_i=0$.
When $\theta_{{\rm{FoV}}, i} \geq 6\; {\rm mrad}$ (i.e., $\theta_{{\rm{FoV}}, i} \geq 5 {\sigma _{{\rm{angle,u}}}}$), as shown in Figs. \ref{fig4}(b), (c), (e), (f), (h) and (i), the analytical curves match the simulation results well.
This indicates that the derived tractable expression (\ref{equ13}) is accurate when $\theta_{{\rm{FoV}}, i}$ is large and can be directly used for performance analysis.

\subsection{Performance Analysis Results}
\begin{figure}
\centering
\includegraphics[width=8.5cm]{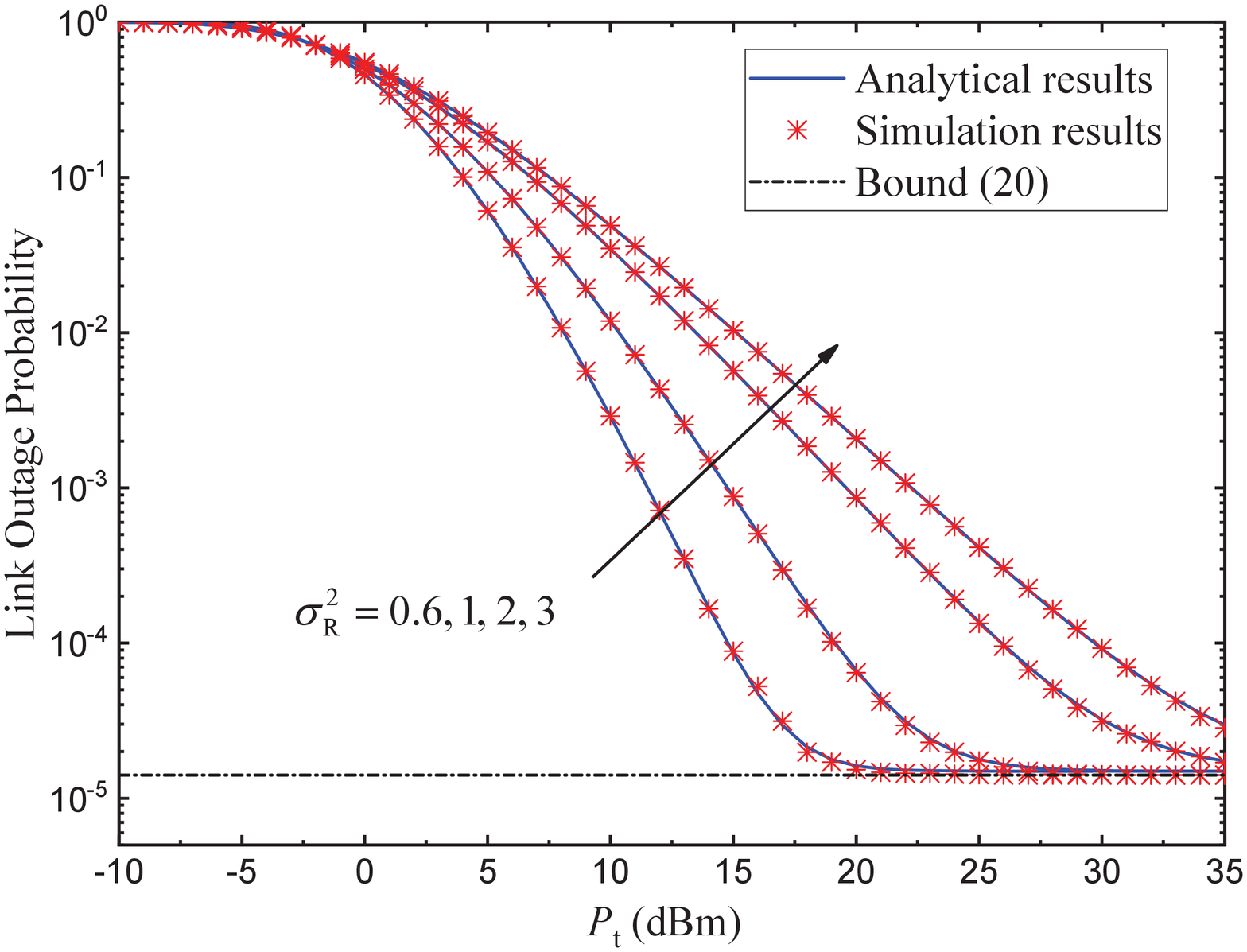}
\caption{Link outage probability for an UU link versus average transmitted optical power under different turbulence conditions when $\theta _{{\rm FoV }} = 8$ mrad, $w_{z} = 2$ m and $Z = 250$ m}
\label{fig5}
\end{figure}

In this subsection, we will examine the accuracy of the derived link outage probability and the end-to-end outage probability under different scenarios.

Fig. \ref{fig5} plots the link outage probability for a UU link versus the average transmitted optical power  $P_{\rm t }$ under different turbulence conditions $\sigma _{{\rm R }}^2 = 0.6, 1, 2, 3$ when $\theta _{{\rm FoV }} = 8$ mrad, $w_{z} = 2$ m and $Z= 250$ m. For comparison, the asymptotic bound (\ref{equ19}) in \emph{Theorem \ref{the2}} is also presented.
For small $P_{\rm t }$, the values of link outage probabilities decrease with the increase of $P_{\rm t }$ or with the decrease of $\sigma _{{\rm R }}^2$.
For large $P_{\rm t }$, the values of link outage probabilities tend to stable values, which are independent of the average transmitted optical power and the turbulence condition. Moreover, the stable values approach the asymptotic bound (\ref{equ19}). This verifies the accuracy of (\ref{equ19}) in \emph{Theorem \ref{the2}}.
Moreover, it should be emphasized that all analytical results in Fig. \ref{fig5} present close agreement with Monte-Carlo simulation results, and the accuracy of the derived expression \eqref{equ18} in \emph{Theorem \ref{the1}} is verified.

\begin{figure}
\centering
\includegraphics[width=8.5cm]{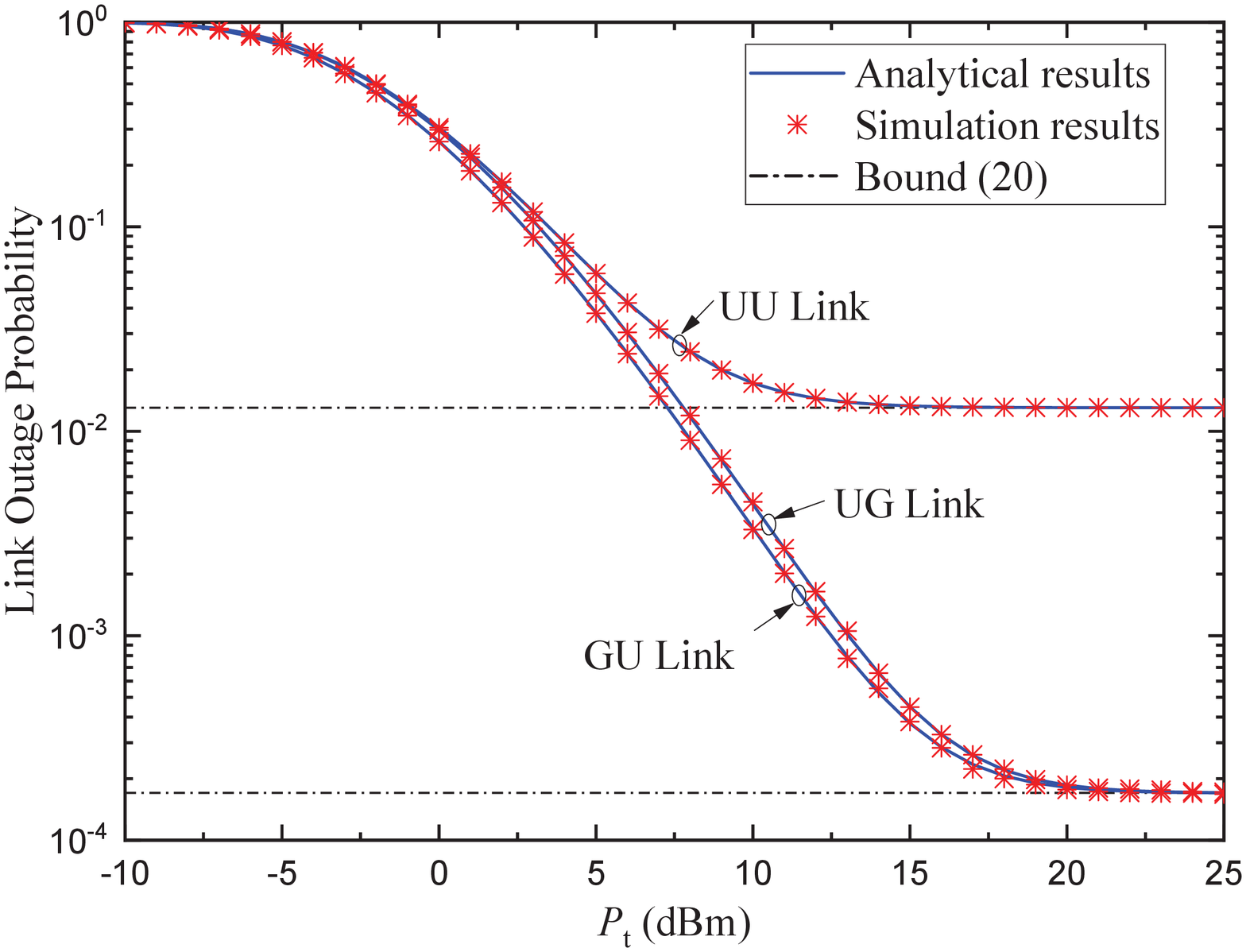}
\caption{Link outage probability versus average transmitted optical power for different kinds of links when $\theta _{{\rm FoV,}i} = 5$ mrad, $\sigma_{{\rm R},i}^2 = 1$, $w_{z,i} = 2$ m and $Z_i = 250$ m}
\label{fig5_1}
\end{figure}

Fig. \ref{fig5_1} compares the link outage performance for different kinds of links when $\theta _{{\rm FoV,}i} = 5$ mrad, $\sigma_{{\rm R},i}^2 = 1$, $w_{z,i} = 2$ m and $Z_i = 250$ m.
In addition to obtaining the similar conclusions in Fig. \ref{fig5}, some other interesting insights can also be found in Fig. \ref{fig5_1}.
As shown in Fig. \ref{fig5_1}, the UU link achieves the maximum link outage probability, and the GU link achieves the comparable performance to the UG link.
Moreover, the GU link achieves the same asymptotic outage performance as the UG link,
and the asymptotic bound of the UU link is always larger than that of the UG/GU link.
These insights verify the analysis in \emph{Remark \ref{rem1}}.

\begin{figure}
\centering
\includegraphics[width=8.5cm]{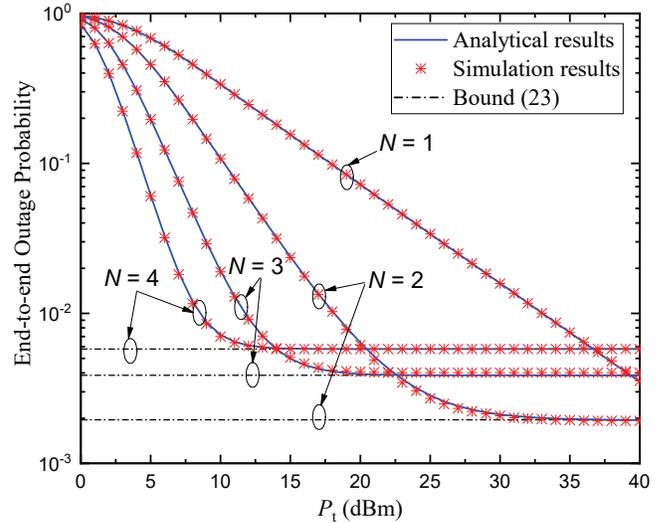}
\caption{End-to-end outage probability versus average transmitted optical power for different numbers of relays when $\theta _{{\rm FoV,}i} = 6$ mrad, $w_{z,i} = 2$ m and the distance between source and destination nodes $Z_{\rm SD} = 2$ km}
\label{fig6}
\end{figure}

Fig. \ref{fig6} shows the end-to-end outage probability versus the average transmitted optical power ${P_{\rm{t}}}$ for different numbers of relays $N$ when ${\theta _{{\rm FoV },i}} = 6$ mrad, $w_{z, i} = 2$ m and the distance between source and destination nodes $Z_{\rm SD} = 2$ km. For comparison, the asymptotic bound (\ref{equ23}) in \emph{Theorem \ref{the4}} is also presented.
In this simulation, the obstacles are not considered, and $N$ UAV nodes are uniformly deployed between the source node and the destination node.
For small ${P_{\rm{t}}}$, the end-to-end outage probability performance improves with the increase of ${P_{\rm{t}}}$ or $N$.
However, at large ${P_{\rm{t}}}$, the outage performance does not improve with ${P_{\rm{t}}}$, but approaches the asymptotic bound \eqref{equ23}, which verifies the accuracy of (\ref{equ23}) in \emph{Theorem \ref{the4}}.
Moreover, the asymptotic bound \eqref{equ23} varies with $N$. Specifically, the asymptotic bound with $N=4$ achieves the largest value, which coincides with the conclusion in \emph{Remark \ref{rem2}}. Similar to Fig. \ref{fig5}, the accuracy of the derived expression of \eqref{equ22} in \emph{Theorem \ref{the3}} is also verified by using simulations.

Fig. \ref{figadd} shows the end-to-end outage probability versus the standard deviation ${\sigma _{{\rm{angle,u}}}}$ with different numbers of relays when ${P_{\rm{t}}} = 20\;{\rm{dBm}}$, ${\theta _{{\rm{FoV,}}i}} = 6\;{\rm{mrad}}$, ${w_{{\rm{z}},i}} = 2\;{\rm{m}}$, and ${Z_{{\rm{SD}}}} = 2\;{\rm{km}}$.
As can be observed, with the increase of ${\sigma _{{\rm{angle,u}}}}$,
the end-to-end outage probability increases sharply and then tends to one.
This indicates that the system performance dramatically degrades with the increase of ${\sigma _{{\rm{angle,u}}}}$. Moreover, a large ${\sigma _{{\rm{angle,u}}}}$ (for example, $\geq$3.4 mrad) will directly result in a complete interruption of the system. Furthermore, when ${\sigma _{{\rm{angle,u}}}}$ is small, the outage probability performance improves with the number of UAV relays. Therefore, the UAVs' orientation fluctuation and the number of relays are two major concerns for practical system design.
In Fig. \ref{figadd}, all analytical results match well with simulation results, which also verifies the accuracy of the derived expression in \eqref{equ22}.

\begin{figure}
\centering
\includegraphics[width=8.5cm]{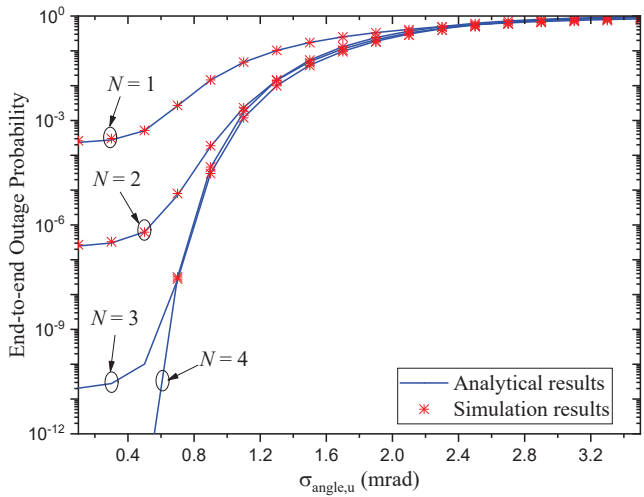}
\caption{End-to-end outage probability versus the standard deviation ${\sigma _{{\rm{angle,u}}}}$ with different numbers of relays when ${P_{\rm{t}}} = 20\;{\rm{dBm}}$, ${\theta _{{\rm{FoV,}}i}} = 6\;{\rm{mrad}}$, ${w_{{\rm{z}},i}} = 2\;{\rm{m}}$, and ${Z_{{\rm{SD}}}} = 2\;{\rm{km}}$}
\label{figadd}
\end{figure}

\subsection{Parameter Optimization Results}
In this subsection, to verify the efficiency of the proposed optimization schemes, the parameter optimization results will be shown.

Fig. \ref{fig7} shows the link outage probability of a UU link versus the normalized FoV ${\phi _{{\rm{FoV}}}}$ for different $P_{\rm{t}}$ values when $Z = 250$ m and ${w_{z}} = 2$ m.
Here, the normalized FoV is defined as ${\phi _{{\rm{FoV}}}} = {\theta _{{\rm{FoV}}}}/{\sigma _{{\rm{angle,u}}}}$, and the beam width $w_{z}$ is adjusted to satisfy (\ref{beam}).
As can be seen, with the increase of $P_{\rm{t}}$, the link outage probability performance improves.
Moreover, all analytical results accurately match simulation results in the entire normalized FoV range, which verifies the accuracy of (\ref{equ18}) in \emph{Theorem \ref{the1}}.
All curves tend to the asymptotic bound in (\ref{equ19}) with the decrease of ${\phi _{{\rm{FoV}}}}$, which indicates the correctness of \emph{Theorem \ref{the2}}.
Furthermore, with the increase of ${\phi _{{\rm{FoV}}}}$, the link outage probability decreases first and then increases.
It can be observed from Fig. \ref{fig7} that, when $P_{\rm t}=0, 5, 10, 15$ and 20 dBm,
the optimal values of ${\phi _{{\rm{FoV}}}}$ (in mrad) are 3.25, 4.25, 5.25, 6.08, and 6.92, respectively.
By solving \eqref{plusFOV1},
we obtain the corresponding asymptotically optimal ${\phi _{{\rm{FoV}}}}$ values as 1.94, 3.67, 4.98, 6.02, and 6.91.
The differences between the optimal ${\phi _{{\rm{FoV}}}}$ in Fig. \ref{fig7} and the asymptotically optimal ${\phi _{{\rm{FoV}}}}$ using \eqref{plusFOV1} are small when $P_{\rm{t}}$ is large.
This indicates that, when $P_{\rm{t}}$ is large, eq. \eqref{plusFOV1} can be directly used to determine the optimal FoV without time-consuming simulations.

\begin{figure}
\centering
\includegraphics[width=8.5cm]{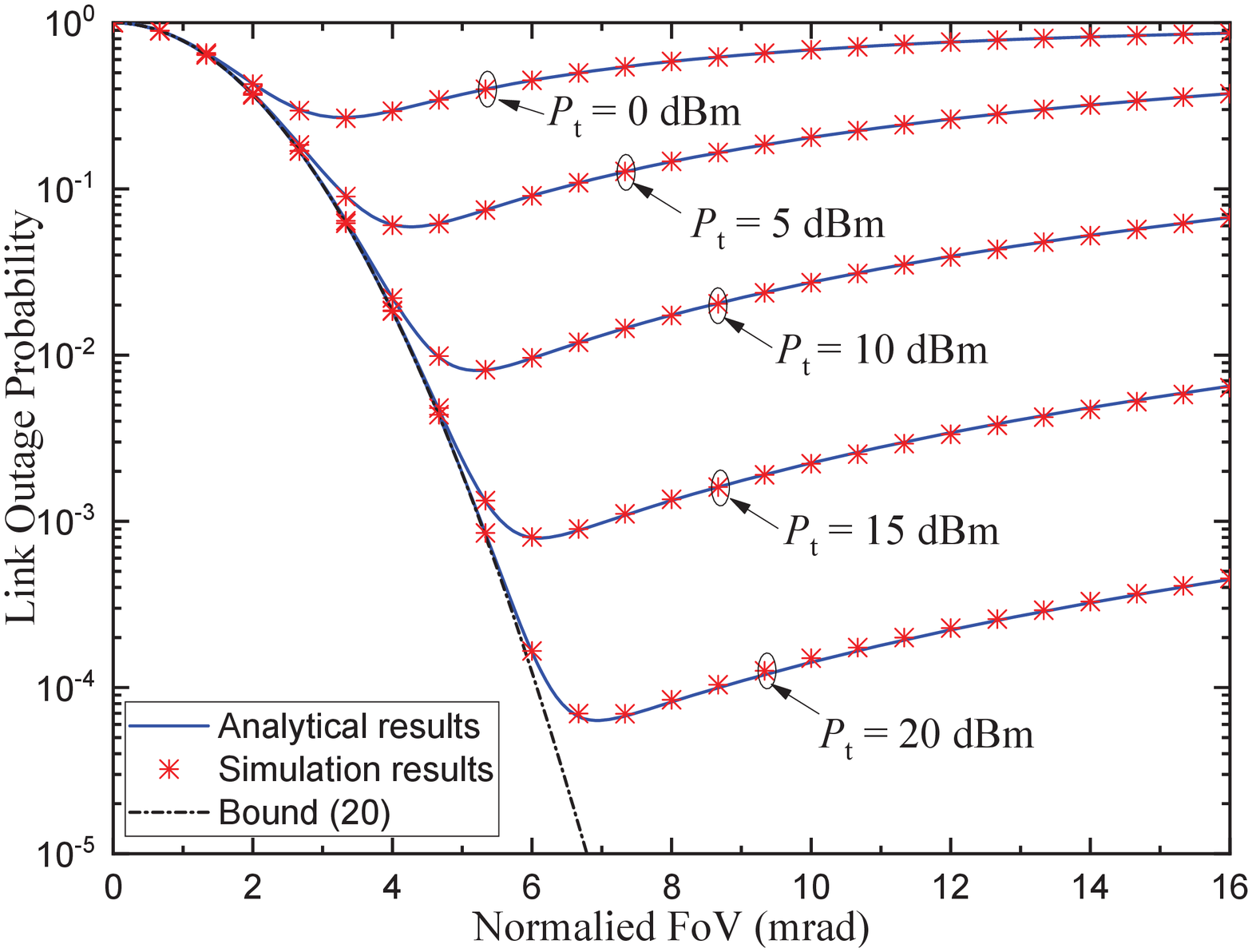}
\caption{Link outage probability of an UU link versus normalized FoV for different $P_{\rm{t}}$ when $Z= 250$ m and ${w_{z}} = 2$ m}
\label{fig7}
\end{figure}

\begin{figure}
  \centering
  \includegraphics[width=8.5cm]{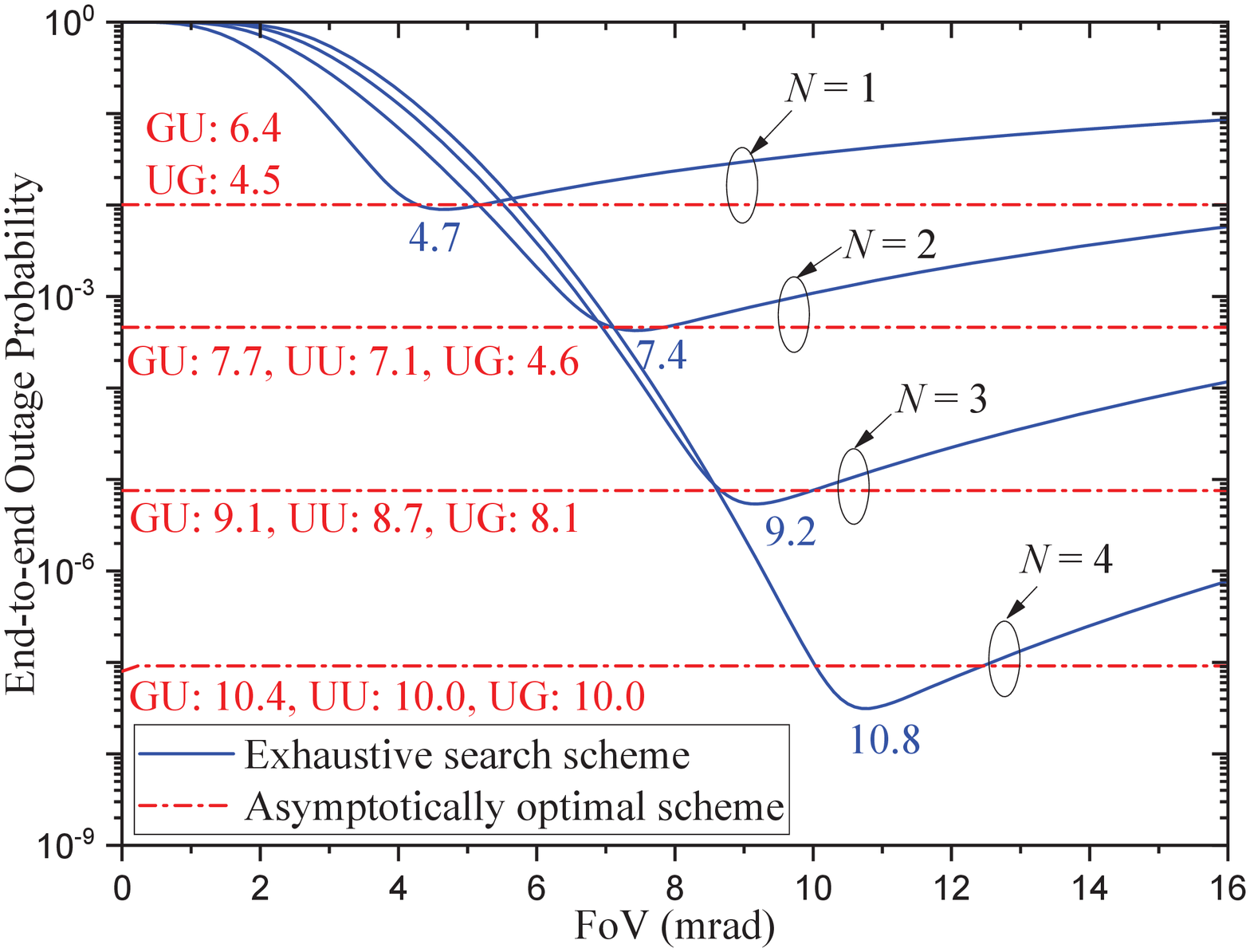}
  \caption{End-to-end outage probability versus FoV with different numbers of relays when ${P_{\rm{t}}} = 20\;{\rm{dBm}}$, ${\sigma _{{\rm{angle,u}}}} = 1.2\;{\rm{mrad}}$, ${\theta _{{\rm{FoV,}}i}} = 6\;{\rm{mrad}}$, ${w_{{\rm{z}},i}} = {\rm{4}}\;{\rm{m}}$, and ${Z_{{\rm{SD}}}} = 2\;{\rm{km}}$}
  \label{fig8}
\end{figure}

Fig. \ref{fig8} shows the end-to-end outage probability versus FoV with different numbers of relays when ${P_{\rm{t}}} = 20\;{\rm{dBm}}$, ${\sigma _{{\rm{angle,u}}}} = 1.2\;{\rm{mrad}}$, ${\theta _{{\rm{FoV,}}i}} = 6\;{\rm{mrad}}$, ${w_{{\rm{z}},i}} = {\rm{4}}\;{\rm{m}}$, and ${Z_{{\rm{SD}}}} = 2\;{\rm{km}}$.
In Fig. \ref{fig8}, the obstacles are not considered, and the UAV relays are placed equidistant.
The exhaustive search scheme and the asymptotically optimal scheme are provided.
In the exhaustive search scheme, the FoVs for the GU, UG and UU links are set to be identical.
In asymptotically optimal scheme, the optimal FoVs for all links are respectively obtained by \eqref{plusFOV1}. It can be seen that when FoV is increased, all curves of the exhaustive search scheme decrease first and then increase, this indicates that an optimal FoV exists for each curve. Specifically, for $N=1,2,3$ and 4, the optimal FoVs in the exhaustive search scheme are 4.7 mrad, 7.4 mrad, 9.2 mrad, and 10.8 mrad, as marked on the figure. The corresponding minimum outage probabilities are $8.95 \times 10^{-3}$, $4.26\times 10^{-4}$, $5.40\times 10^{-6}$, and $3.13\times 10^{-8}$, respectively.
Moreover, the FoVs of different links for the asymptotically optimal scheme by using \eqref{plusFOV1} are also provided on the figure.
As seen, the asymptotically optimal scheme can achieve comparable outage performance to the exhaustive search scheme, which verifies the accuracy of \eqref{plusFOV1}.

\begin{figure}
\centering
\includegraphics[width=8.5cm]{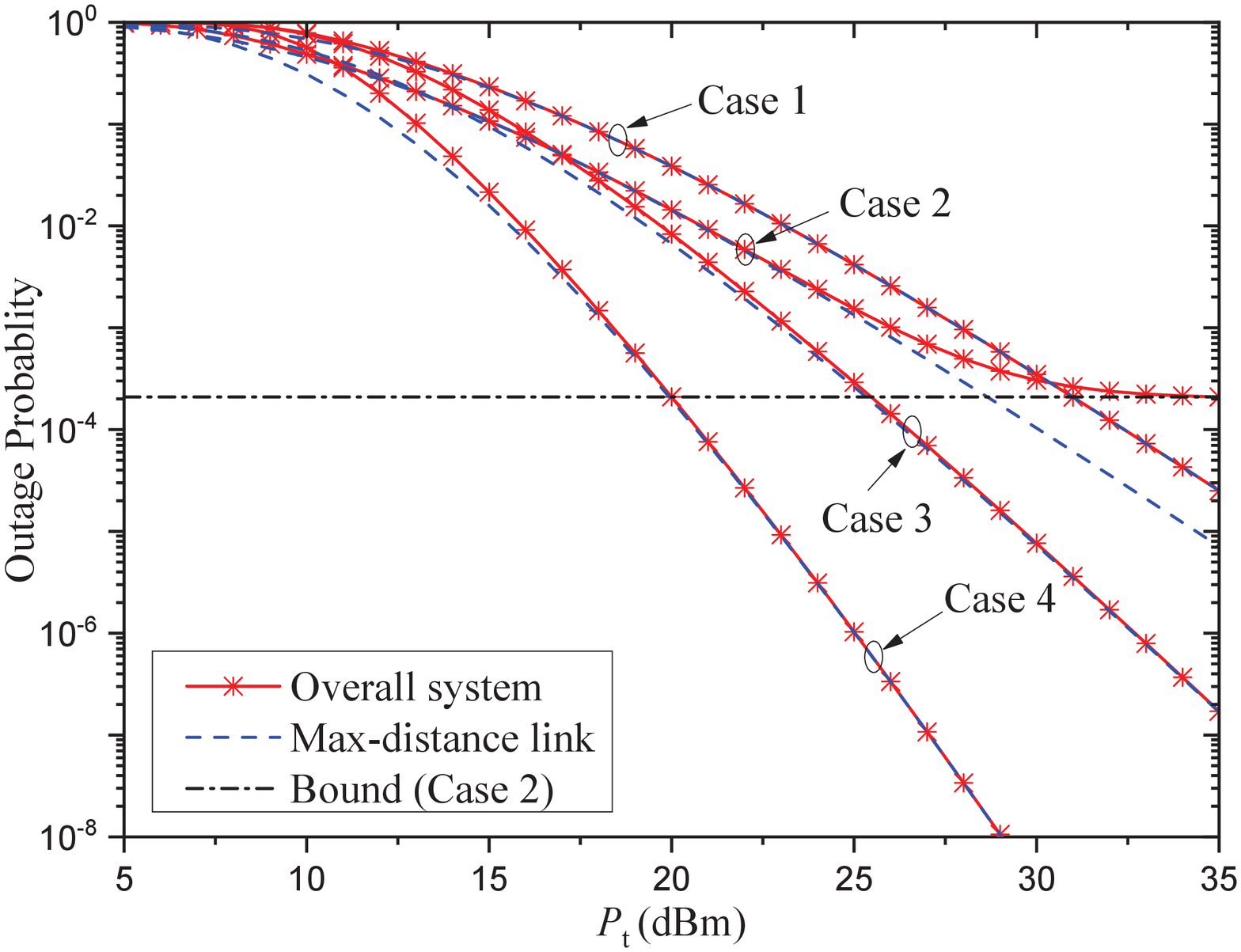}
\caption{Outage probability versus average link transmit power for different scenarios when ${Z_{\rm SD}} = 2 $ km and ${w_{z,i}} = 4$ m}
\label{fig9}
\end{figure}

To verify whether the link with maximum distance (i.e., max-distance link) dominates the end-to-end outage performance, Fig. \ref{fig9} shows the outage probability versus average transmit power $P_{\rm t}$ for different scenarios when ${Z_{\rm SD}} = 2 $ km and ${w_{z,i}} = 4$ m. In this simulation, the beam width satisfies $\zeta _i^2 > {\beta _i},i = 1,...,{N} + 1$. Four cases are considered, as shown in Table \ref{tab1}.
Note that the asymptotic bounds for Cases 1 and 3 are $1.39 \times 10^{-11}$, the asymptotic bound for Case 4 is $2.78 \times 10^{-11}$, which are too small and not plotted in the figure.
For cases with large FoV (i.e., Cases 1, 3 and 4) in Fig. \ref{fig9}, the outage probabilities of the max-distance link converge to that of the overall system when $P_{\rm t}$ is large (for example, when $P_{\rm t} \ge 12.5$ dBm for Case 1, when $P_{\rm t} \ge 27.5$ dBm for Case 3, or when $P_{\rm t} \ge 20$ dBm for Case 4), which verifies that the link with the maximum distance dominates the overall system performance, just as analyzed in \eqref{MDL}.
However, for the case with small FoV (i.e., Case 2), the outage probability of the max-distance link does not match with that of the overall system at large $P_{\rm t}$. In this case, the AoA fluctuation has a strong effect on outage performance.
Moreover, the performance of Case 2 is better than that of Case 1 before reaching the asymptotic bound \eqref{equ23} since larger FoV introduces more background noise.

\begin{table}
\renewcommand{\arraystretch}{1}
\caption{The setup of four cases in Fig. \ref{fig9}}
\label{tab1}
\centering
\begin{tabularx}{8.8cm}{|p{0.8cm}|p{1.1cm}|p{3.6cm}|X|}
\hline
\textbf{Cases} &  \textbf{Number of relays} & \textbf{Link distance (km)} & \textbf{FoV (mrad)} \\
\hline
Case 1 & 2 & $({Z_1},{Z_2},{Z_3}) = (0.5,0.5,1)$ & 12\\
\hline
Case 2 & 2 & $({Z_1},{Z_2},{Z_3}) = (0.5,0.5,1)$ & 7\\
\hline
Case 3 & 2 & $({Z_1},{Z_2},{Z_3}) = (0.5,0.7,0.8)$ & 12\\
\hline
Case 4 & 3 & $({Z_1},{Z_2},{Z_3},{Z_4}) = (0.4,0.5,0.5,0.6)$ &  12\\
\hline
\end{tabularx}
\end{table}

\begin{figure}
\centering
\includegraphics[width=9cm]{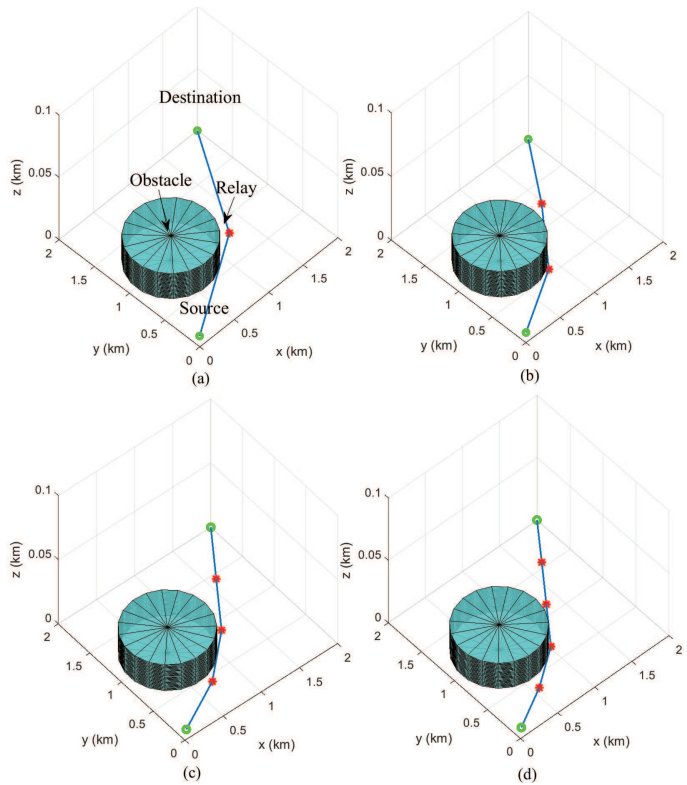}
\caption{Optimal relay deployment for the UAV-based FSO relaying system with one cylindrical obstacle. (a) one relay, (b) two relays, (c) three relays, (d) four relays}
\label{fig10}
\end{figure}

After optimizing the UAVs' locations, Fig. \ref{fig10} and Fig. \ref{fig11} show the deployment of UAVs under different scenarios.
Specifically, Fig. \ref{fig10} plots optimal relay deployment for the UAV-based FSO relaying system with one cylindrical obstacle.
In the simulation, the coordinates of the source and the destination are (0.1 km, 0.1 km) and (2 km, 2 km).
The obstacle is centered at (0.6 km, 1 km) with radius 0.5 km.
From Fig. \ref{fig10}(a) to Fig. \ref{fig10}(d), the optimal locations of UAVs and link distances are listed in Table \ref{tab2}.
By adding another cylindrical obstacle at (1.6 km, 1.2 km) with radius 0.2 km, Fig. \ref{fig11} shows the optimal relay deployment for the FSO relaying system. In this case, the optimal relay locations and link distances are provided in Table \ref{tab3}.
As can be seen from the above results, through optimizing the relays' deployment, the feasible positions of relays can always be found, and all link distances tend to be approximately equal.

\begin{table}
\renewcommand{\arraystretch}{1}
\caption{The optimal relay locations and link distances in Fig. \ref{fig10}}
\label{tab2}
\centering
\begin{tabularx}{9cm}{|p{1.2cm}|X|X|}
\hline
\textbf{Subfigures} &  \textbf{Relay locations (in km)} & \textbf{Link distances (in km)} \\
\hline
Fig. \ref{fig10}(a) & (1.2712, 0.8288) & 1.3795, 1.3795 \\
\hline
Fig. \ref{fig10}(b) & (0.8851, 0.5614), (1.4689, 1.2602) & 0.9106, 0.9106, 0.9106\\
\hline
Fig. \ref{fig10}(c) & (0.7056, 0.4193), (1.2111, 0.8809), (1.6056, 1.4404)& 0.6846, 0.6846, 0.6846, 0.6846\\
\hline
Fig. \ref{fig10}(d) & (0.5777, 0.3651), (1.0193, 0.6868), (1.3472, 1.1238), (1.6745, 1.5612) & 0.5463, 0.5463, 0.5463, 0.5463, 0.5463\\
\hline
\end{tabularx}
\end{table}

\begin{figure}
\centering
\includegraphics[width=9cm]{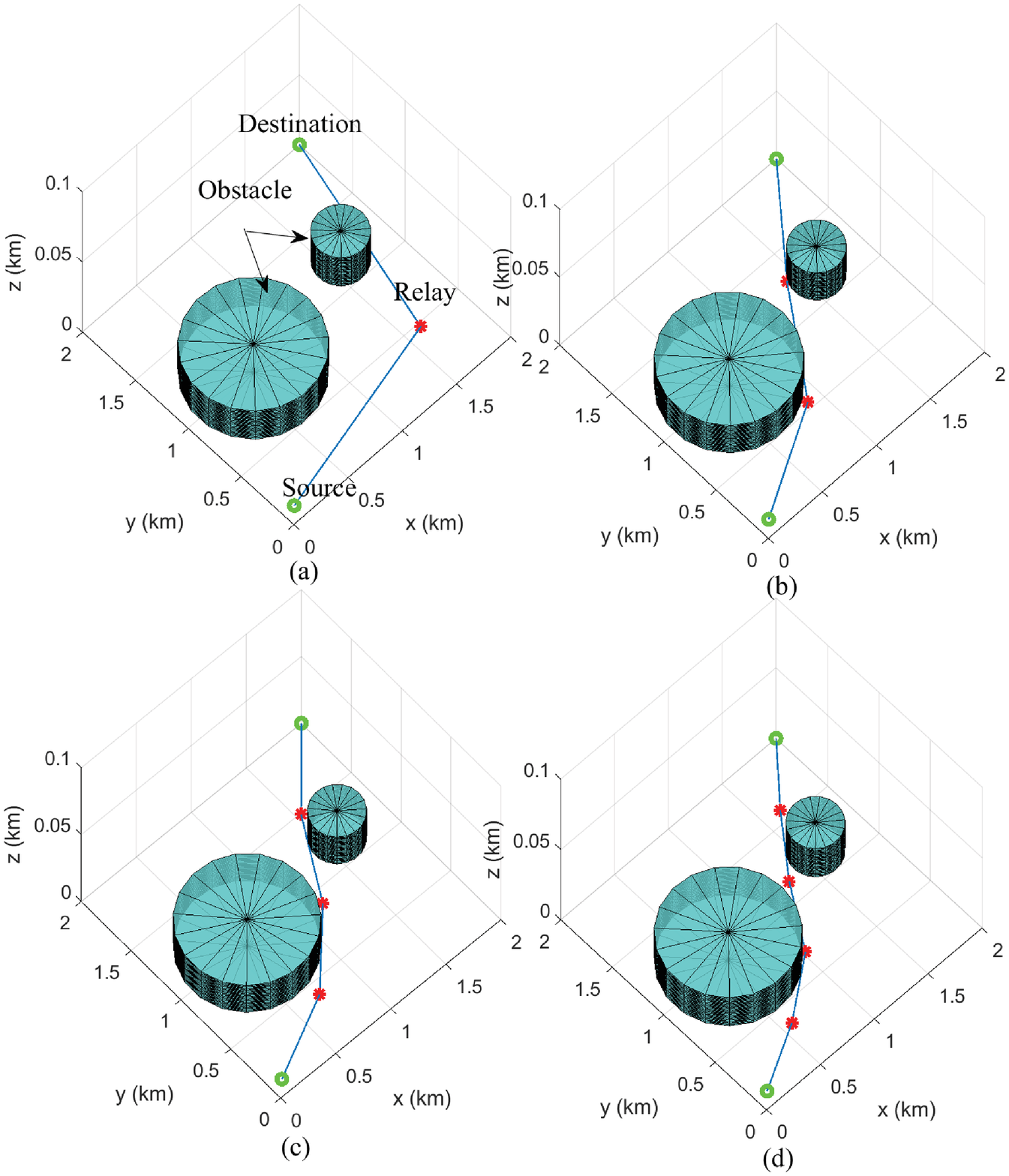}
\caption{Optimal relay deployment for the UAV-based FSO relaying system with two cylindrical obstacles. (a) one relay, (b) two relays, (c) three relays, (d) four relays}
\label{fig11}
\end{figure}

\begin{table}
\renewcommand{\arraystretch}{1}
\caption{The optimal relay locations and link distances in Fig. \ref{fig11}}
\label{tab3}
\centering
\begin{tabularx}{9cm}{|p{1.2cm}|X|X|}
\hline
\textbf{Subfigures} &  \textbf{Relay locations (in km)} & \textbf{Link distances (in km)} \\
\hline
Fig. \ref{fig11}(a) & (1.2712, 0.8288) & 1.3795, 1.3795 \\
\hline
Fig. \ref{fig11}(b) & (0.8853, 0.5611), (1.4676, 1.2612) & 0.9106, 0.9106, 0.9106\\
\hline
Fig. \ref{fig11}(c) & (0.7057, 0.4191), (1.2109, 0.8811), (1.6055, 1.4405)& 0.6846, 0.6846, 0.6846, 0.6846\\
\hline
Fig. \ref{fig11}(d) & (0.5731, 0.3741), (1.0221, 0.6837), (1.3437, 1.1259), (1.6729, 1.5626) & 0.5468, 0.5454, 0.5468, 0.5468, 0.5462\\
\hline
\end{tabularx}
\end{table}

\section{Conclusions}
\label{section6}
This paper investigated the channel modelling, outage probability analysis, and parameter optimization for UAV-based FSO relaying systems.
A tractable and accurate channel model was established by taking into account atmospheric loss, atmospheric turbulence, pointing error, and link interruption due to AoA fluctuation.
Based on the channel model, closed-form expressions were obtained for link outage probability and end-to-end outage probability.
Subsequently, the asymptotic outage performance bounds were also investigated.
Moreover, the beam width, FoV and UAVs' locations were optimized.
Numerical results verify the accuracy of the derived theoretical expressions and the efficiency of the proposed optimization schemes.
The obtained theoretical expressions enable designers to evaluate outage performance rapidly without time-intensive simulations.
The derived parameter optimization results can help determine the optimal available parameter choices when designing the UAV-based FSO systems.

\numberwithin{equation}{section}
\appendices
\section{Proof of \emph{Theorem \ref{the0}}}
\label{appa}
\renewcommand{\theequation}{A.\arabic{equation}}
By using \eqref{equ3}, \eqref{equ4}, \eqref{equ8} and expressing ${K_n}( \cdot )$ in terms of the Meijer's G-function \cite{add1}, we can obtain the PDF of $h_i^{'}$ as \cite{ShM}
\begin{eqnarray}
f_{h_{i}^{\prime}}\left(h_{i}^{\prime}\right) \!\!\!\!&=&\!\!\!\! \frac{{{\alpha _i}{\beta _i}\zeta _i^2}}{{{A_i}h_i^{({\rm{l}})}\Gamma ({\alpha _i})\Gamma ({\beta _i})}}\nonumber\\
&\times&\!\!\!\! G_{1,3}^{3,0}\!\!\left[ {\frac{{{\alpha _i}{\beta _i}}}{{{A_i}h_i^{{\rm{(l}})}}}{h_{i}^{\prime}}\! \left|\!\! {\begin{array}{*{20}{c}}
{\zeta _i^2}\\
{\zeta _i^2 - 1,{\alpha _i} - 1,{\beta _i} - 1}
\end{array}} \right.} \!\!\!\! \right]\!.
\label{equ9}
\end{eqnarray}

By substituting \eqref{equ11} and \eqref{equ9} into \eqref{equ12}, we derive the PDF of ${h_i}$ as
\begin{eqnarray}
{f_{{h_i}}}\!\! \left( {{h_i}} \right) \!\!&=&\!\! \exp \left( \!\!{ - \frac{{\theta _{{\rm{FoV}},i}^2}}{{2m_i\sigma _{{\rm{angle}},{\rm{u}}}^{\rm{2}}}}} \right)\int_0^\infty  {\frac{1}{{h_i^\prime }}\delta \left( {\frac{{{h_i}}}{{h_i^\prime }}} \right){f_{h_i^\prime }}\left( {h_i^\prime } \right){\rm{d}}h_i^\prime } \nonumber \\
 &+&\!\! \left[ {1\! -\! \exp \!\left( { \!-\! \frac{{\theta _{{\rm{FoV}},i}^2}}{{2m_i\sigma _{{\rm{angle}},{\rm{u}}}^{\rm{2}}}}} \right)} \right]\nonumber\\
 &\times&\int_0^\infty  {\!\!\frac{1}{{h_i^\prime }}\delta \left( {\frac{{{h_i}}}{{h_i^\prime }}\!\! -\!\! 1} \right){f_{h_i^\prime }}\left( {h_i^\prime } \right){\rm{d}}h_i^\prime } .
 \label{PDF1}
\end{eqnarray}
With $\delta(ax)=\delta(x) / {|a|} $ and $\int_{ - \infty }^\infty  f (a)\delta (x - a){\rm d}a = f(x)$ \cite{c1}, eq. (\ref{equ13}) can be derived.

\section{Proof of \emph{Theorem \ref{the1}}}
\label{appb}
\renewcommand{\theequation}{B.\arabic{equation}}
By \eqref{equ13} and \eqref{equ16}, $p_i$ is given by
\begin{eqnarray}
{p_i} \!\!\!\!\! &=&\!\!\!\!\! \exp\left( { - \frac{{\theta _{{\rm{FoV}},i}^2}}{{2m_i\sigma _{{\rm{angle,u}}}^{\rm{2}}}}}\! \right)\!\!
\underbrace {\int_0^{{h_{{\rm{th}},i}}}\delta ({h_i}) {\rm d} h_i}_{=1} \nonumber\\
&+& \!\!\!\!\! \left[\! {1 \!- \!\exp\! \left(\! \!{ - \frac{{\theta _{{\rm{FoV}},i}^2}}{{2m_i\sigma _{{\rm{angle,u}}}^{\rm{2}}}}}\! \right)}\!\! \right]\frac{{{\alpha _i}{\beta _i}\zeta _i^2}}{{{A_i}h_i^{({\rm{l}})}\Gamma ({\alpha _i})
\Gamma ({\beta _i})}} \nonumber \\
&\times&\! \!\!\!\!\!\underbrace {\int_0^{{h_{{\rm{th}},i}}} \!\!\! G_{1,3}^{3,0}\!\!\left[\! {\frac{{{\alpha _i}{\beta _i}}}{{{A_i}h_i^{{\rm{(l}})}}}{h_i} \left|\!\! {\begin{array}{*{20}{c}}
{\zeta _i^2}\\
{\zeta _i^2 \!-\! 1,{\alpha _i} \!-\! 1,{\beta _i} \!-\! 1}
\end{array}} \right.}\!\!\!\! \right]{\rm d} h_i}_{\triangleq I}.
\label{26}
\end{eqnarray}

By (26) in \cite{add1}, $I$ in (\ref{26}) can be further written as
\begin{equation}
I \!=\!{h_{{\rm{th,}}i}} G_{2,4}^{3,1}\!\left[\! {\frac{{{\alpha _i}{\beta _i}}}{{{A_i}h_i^{({\rm{l)}}}}}{h_{{\rm{th,}}i}}\left| {\begin{array}{*{20}{c}}
{0,\zeta _i^2 }\\
{\zeta _i^2\!-\!1,{\alpha _i}\!-\!1,{\beta _i}\!-\!1,-1}
\end{array}} \right.}\!\!\!\! \right]\!\!.\!\!\!\!
\label{27}
\end{equation}
Then, by (9.31.5) in \cite{FSO9}, eq. (\ref{27}) can be written as
\begin{equation}
I \!=\!\frac{{{A_i}h_i^{({\rm{l)}}}}}{{{\alpha _i}{\beta _i}}} G_{2,4}^{3,1}\left[ {\frac{{{\alpha _i}{\beta _i}}}{{{A_i}h_i^{({\rm{l)}}}}}{h_{{\rm{th,}}i}}\left| {\begin{array}{*{20}{c}}
{1,\zeta _i^2+1 }\\
{\zeta _i^2,{\alpha _i},{\beta _i},0}
\end{array}} \right.} \right].
\label{28}
\end{equation}
Submitting (\ref{28}) and (\ref{equ17}) into (\ref{26}), we obtain (\ref{equ18}).

\section{Proof of \emph{Theorem \ref{the2}}}
\label{appc}
\renewcommand{\theequation}{C.\arabic{equation}}
According to (\ref{equ17}), when the transmit power ${P_{\rm{t}}}$ tends to infinity, we have
 \begin{equation}
\mathop {\lim }\limits_{{P_{\rm{t}}} \to \infty } {h_{{\rm{th,}}i}} = \mathop {\lim }\limits_{{P_{\rm{t}}} \to \infty }\sqrt {\frac{{{\Upsilon _{{\rm{th}}}}\sigma _{{\rm{n,}}i}^{\rm{2}}}}{2R^2{P_{\rm{t}}^2}}}=0.
\label{30}
\end{equation}

In (\ref{equ18}), the Meijer's G-function cannot provide intuitive insights on the behavior of ${h_{{\rm{th,}}i}}$.
For a large transmit power $P_{\rm t}$, the outage probability is dominated by the behavior of the PDF near the origin.
Therefore, by employing (07.34.06.0006.01) in \cite{FSO12} or \cite{FSO11}, we can approximate (\ref{equ18}) as
\begin{eqnarray}
&&\!\!\!\!\!\!\!\!\!\!\!\!\!\!\!\!\!\!\!\!\!\!\!\! {p_i}\approx \exp\! \left(\!\! { - \frac{{\theta _{{\rm{FoV}},i}^2}}{{2m_i\sigma _{{\rm{angle,u}}}^{\rm{2}}}}}\!\! \right)\! +\!\! \left[\! {1\! - \!\exp \left( \!\!{ - \frac{{\theta _{{\rm{FoV}},i}^2}}{{2m_i\sigma _{{\rm{angle,u}}}^{\rm{2}}}}} \right)}\!\! \right]\nonumber \\
&&\!\!\!\!\!\!\!\!\!\!\!\!\!\!\!\!\!\! \times \frac{{\zeta _i^2\Gamma ({\alpha _i} \!-\! {\kappa _i}){b_i}}}{{\Gamma ({\alpha _i})\Gamma ({\beta _i}){\kappa _i}}}\!{\left(\! \! {\frac{{{\alpha _i}{\beta _i}}}{{{A_i}h_i^{({\rm{l}})}}}}\!\! \right)^{\!\!{\kappa _i}}}\!\!\!(h_{{\rm{th,}}i})^{\kappa _i} ,{\kappa _i}\! = \!\min (\zeta _i^2,{\beta _i}),
\label{equ20}
\end{eqnarray}
where ${b_i} = 1/(\zeta _i^2 - {\beta _i})$ if $\zeta _i^2 > {\beta _i}$, and ${b_i} = \Gamma ({\beta _i} - \zeta _i^2)$ if $\zeta _i^2 < {\beta _i}$. Because the small-scale turbulence eddies parameter $\beta_i  > 1$ \cite{FSO14}, the inequality ${\kappa _i} > 1$ always holds.

By (\ref{30}), the second term in \eqref{equ20} tends to zero when $P_{\rm t}$ tends to infinity. Therefore,  \emph{Theorem {\ref{the2}}} holds.


\end{document}